\documentclass{elsarticle}
\usepackage{graphicx}
\usepackage{verbatim}
\usepackage[latin1]{inputenc}
\usepackage{xcolor}
\usepackage{color}
\usepackage{amsmath}
\usepackage{amssymb}
\usepackage{amsthm}
\usepackage{algorithmic}
\usepackage{algorithm}
\usepackage{relsize}
\usepackage{tikz}

\usetikzlibrary{positioning}
\usetikzlibrary{matrix}
\usetikzlibrary{calc}
\usetikzlibrary{shapes}
\usetikzlibrary{arrows}
\usetikzlibrary{decorations}
\usetikzlibrary{snakes}

\newtheorem{thm}{Theorem}[section]

\newtheorem{lem}[thm]{Lemma}
\newtheorem{prop}[thm]{Proposition}
\theoremstyle{definition}
\newtheorem{defn}[thm]{Definition}
\theoremstyle{remark}

\newcommand{\La}{\mathcal{L}}

\newcommand{\Shift}[1]{\delta_{#1}}
\newcommand{\Mass}[1]{\mu_{#1}}
\newcommand{\rank}{r}
\newcommand{\modRandom}{\textbf{ModRandom}}
\newcommand{\LeftTree}{{\mathcal{T}_L}}
\newcommand{\RightTree}{{\mathcal{T}_R}}
\newcommand{\addToTree}{\textbf{AddToTree}}

\newcommand{\NtForbs}[1]{\overline{\Pond}_{\Forb}(#1)}

\newcommand{\StepByStep}{\text{\bf StepByStep}}
\newcommand{\R}{\mathcal{R}}
\newcommand{\Gram}{\mathcal{G}}
\newcommand{\NtSet}{\mathcal{N}}
\newcommand{\Nt}{N}
\newcommand{\T}{t}
\newcommand{\Voc}{\Sigma}

\newcommand{\Axiom}{\mathcal{S}}
\newcommand{\ProdRules}{\mathcal{P}}
\newcommand{\Lang}[1]{\mathcal{L}(#1)}

\newcommand{\LangImm}[1]{\mathcal{L}^{\lhd}(#1)}
\newcommand{\BigO}[1]{\mathcal{O}(#1)}
\newcommand{\Def}[1]{\textbf{#1}}
\newcommand{\UnG}[2]{#1\,|\,#2}
\newcommand{\PrG}[2]{#1\,.\,#2}
\newcommand{\Production}{\rightarrow}
\newcommand{\Derive}{\Rightarrow}
\newcommand{\DerPol}{\phi}
\newcommand{\DerPolLR}{\DerPol_{L}}

\newcommand{\AlgW}{\mathcal{A}}

\newcommand{\Remark}[1]{\vspace{.3em}\noindent{{\bf Remark: } #1}\vspace{.3em}}
\newcommand{\Prob}[1]{\mathbb{P}(#1)}
\newcommand{\Expect}[1]{\mathbb{E}(#1)}
\newcommand{\Forb}{\mathcal{F}}
\newcommand{\Pond}{\pi}
\newcommand{\Eqdef}{:=}
\newcommand{\GramW}{\mathcal{G}_\Pond}

\newcommand{\Draw}[1]{\text{draw}(#1)}
\newcommand{\OrderRel}{\preccurlyeq}

\pgfdeclarelayer{background}
\pgfdeclarelayer{foreground}
\pgfsetlayers{background,main,foreground}

\begin{document}
\begin{frontmatter}
\title{Non-redundant random generation algorithms for weighted context-free grammars}%
\author[andy]{Andy Lorenz}
\author[yann]{Yann Ponty\corref{cor1}}
\ead{yann.ponty@lix.polytechnique.fr}

\address[andy]{Mathematics Departement\\
Denison University\\Granville, USA}%
\address[yann]{CNRS/Inria AMIB\\
Ecole Polytechnique\\ Palaiseau, France}%
\cortext[cor1]{Corresponding author}
\begin{abstract}
	We address the non-redundant random generation of $k$ words of length $n$ in a context-free language.
	Additionally, we want to avoid a predefined set of words. We study a rejection-based approach, whose worst-case time complexity is shown to
    grow exponentially with $k$ for some specifications and in the limit case of a coupon collector. We propose two
	algorithms respectively based on the recursive method and on an unranking approach.
    We show how careful implementations of these algorithms allow for a non-redundant generation of $k$ 	words of length $n$ in $\mathcal{O}(k\cdot n\cdot \log{n})$ arithmetic operations, after a precomputation of
	$\Theta(n)$ numbers. The overall complexity is therefore dominated by the generation of $k$ words, and the non-redundancy comes at a negligible cost.
	\end{abstract}


\end{frontmatter}

\section{Introduction}
	The random generation of combinatorial objects has many direct applications in areas
	ranging from software testing~\cite{Den06} to bioinformatics~\cite{PoTeDe06}.
	It can help formulate conjectures on the average-case complexity of algorithms~\cite{Bassino2009},
	raises new fundamental mathematical questions, and motivates new developments on
    its underlying objects. These include, but are not limited to, generating
	functionology, arbitrary precision arithmetics and bijective combinatorics. Following the
	\emph{recursive} framework introduced by Wilf~\cite{wilf77},
	very elegant and general algorithms for the uniform random generation have been designed~\cite{flajoletcalculus} and implemented.
	Many optimizations of this approach have been developed, using specificities of certain classes of
	combinatorial structures~\cite{Gol95}, or floating-point arithmetics~\cite{DeZi99}. More recently, Boltzmann sampling~\cite{fullboltz}, 	
 an algebraic approach based on analytic combinatorics, has drawn much attention, mostly owing to its minimal
 memory consumption and its intrinsic theoretical elegance.

	For many applications, it is necessary to depart from \emph{uniform} models~\cite{DiLa03,Brlek2006}.
	A striking example lies in a recent paradigm for the \emph{in silico} analysis of the folding of Ribo-Nucleic Acids (RNAs).
	Instead of trying to predict a conformation of minimal free-energy, current approaches tend to focus on the \emph{ensemble
	properties} of realizable conformations, assuming a Boltzmann probability distribution~\cite{DiLa03} on the entire set of conformations.
	Random generation is then performed, and complex structural features are evaluated in a
	statistical manner. In order to capture such features, a general non-uniform scheme was introduced by Denise
	\emph{et al}~\cite{deniserandom}, based on the concept of \emph{weighted context-free grammars}.
    Recursive random generation algorithms were derived, with time and space complexities equivalent to that
    observed within the uniform distribution~\cite{flajoletcalculus}. This initial work was later completed toward
    general decomposable classes~\cite{Denise2010} and  a Boltzmann weighted sampling scheme, used as a preliminary
    step within a rejection-based algorithm for the multidimensional sampling of languages~\cite{Bodini2010}.

	In a weighted probability distribution, the probability ratio between the most and least frequent words
	typically grows exponentially on the size of the generated objects. Therefore a typical set of independently
    generated objects may feature a large number of copies of the heaviest (i.e. most probable) objects. This redundancy, which can be
    useful in some context, such as the estimation the probability of each sample from its frequency, is completely uninformative
    in the context of weighted random generation, as the exact probability of any sampled object can be derived in a
    straightforward manner. Consequently it is a natural question to address
	the \Def{non-redundant random generation} of combinatorial objects, i.e. the generation of a
	set of \Def{distinct} objects.

    The non-redundant random generation has, to the best of our knowledge, only  been addressed indirectly
    through the introduction of the {\tt PowerSet} construct by Zimmermann~\cite{ZimPowerset95}.
    An algorithm in $\Theta(n^2)$ arithmetic operations, or a
	practical $\Theta(n^4)$ complexity in this case, was derived for recursive decomposable structures.
	The absence of redundancy in the generated set of structures was achieved respectively through \emph{rejection}
	or an \emph{unranking} algorithms. Unfortunately, these approaches do not transpose well to the case of 
  weighted languages. Indeed, the former rejection algorithm may have exponential time-complexity in the average-case, as is shown later in the article. The unranking approach benefits from recent contributions by Martinez and Molinero~\cite{Martinez00ageneric}, who gave
  generic unranking procedures for labeled combinatorial classes, generalized by Weinberg and Nebel~\cite{WeinbergNebel} to rule-weighted 
  context-free grammars. However, the latter algorithm is restricted to integral weights, and requires a transformation of the grammar which may impact 
  its complexity. Furthermore, the question of figuring out a rank which avoids a set of words was completely ignored by these works.
    
	In this paper, we address the non-redundant generation of words from a context-free language.
    We remind or introduce in Section~\ref{sec:notations} some concepts and definitions related to weighted languages, and
    define our objective.
    In Section~\ref{sec:rejet}, we analyze the shortcomings of a naive rejection approach.
    We show that, although well-suited for the uniform
	distribution, the rejection approach may lead to prohibitive average-case complexities in the case of degenerate grammars, large sets
	of forbidden words, large weights values, or large sets of generated words.
    Then, in Section~\ref{sec:unifForb}, we introduce the concept of immature words, which allows us to rephrase the random generation process as a
    \emph{step-by-step} process. The resulting algorithm is based on the recursive method, coupled with a custom data structure to
    perform a generation of $k$ sequences of length $n$
	at the cost of $\mathcal{O}(k\cdot n\log(n))$ arithmetic operations after a precomputation in $\Theta(n)$
	arithmetic operations. We also propose in Section~\ref{sec:unrank} an unranking algorithm for weighted grammars which,
    coupled with a dedicated data structure that stores and helps avoid any forbidden word, also yields a $\mathcal{O}(k\cdot n\log(n))$ algorithm after $\Theta(n)$ arithmetic operations.
    We conclude in Section~\ref{sec:conclusion} with a summary of our propositions and results, and outline some perspectives and open questions.

\section{Notations and concepts}\label{sec:notations}
	\subsection{Context-free grammars}
	Let us remind, for the sake of completeness, some basic language-theoretic definitions.
	A \Def{context-free grammar} is a 4-tuple $\Gram=(\Voc,\NtSet,\ProdRules,\Axiom)$ where
	\begin{itemize}
		\item $\Voc$ is the alphabet, i.e. a finite set of terminal symbols.
		\item $\NtSet$ is a finite set of non-terminal symbols.
		\item $\ProdRules$ is the finite set of production rules, each of the form $\Nt\Production X$,
		for $\Nt\in\NtSet$ any non-terminal and $X\in \{\Voc\cup\NtSet\}^*$.
		\item $\Axiom$ is the \Def{axiom} of the grammar, i. e. the initial non-terminal.
	\end{itemize}
	A grammar $\Gram$ is then said to be in \Def{Binary Chomsky Normal Form} (BCNF) iff each of its non-terminals $\Nt\in\NtSet$ is productive and can only be derived using a limited number of production rule (two for \emph{union} type non-terminals, and one otherwise):
	\begin{itemize}
		\item Product type: $\Nt \Production \PrG{\Nt'}{\Nt''}$ with $\Nt',\Nt''\in\NtSet$;
		\item Union type: $\Nt \Production \UnG{\Nt'}{\Nt''}$ with $\Nt',\Nt''\in\NtSet$;
		\item Terminal type: $\Nt \Production \T$ with $\T\in\Voc$;
		\item Epsilon type: $\Nt \Production \varepsilon$, iff $\Nt$ cannot be derived from self-referential non-terminals.
	\end{itemize}
In the following, it will be assumed that the input grammar is given in BCNF. This restriction does not cause any loss of 
generality or performance, as it can be shown that any Chomsky Normal Form grammar can be transformed in linear time into an equivalent 
BCNF grammar, having equal number of rule up to a constant ratio.

	Let $\Lang{\Nt}$ be the \Def{language} associated to $\Nt\in\Voc$ within a grammar $\Gram$, i.e. the set of words composed of
    terminal symbols that can be generated starting from $\Nt$ through a sequence of derivations. One has
    \begin{equation} \Lang{\Nt} = \left\{  \begin{array}{cl}
      \Lang{\Nt'}\times\Lang{\Nt''} & \text{If }\Nt \Production \PrG{\Nt'}{\Nt''}\\
      \Lang{\Nt'}\cup\Lang{\Nt''} & \text{If }\Nt \Production \UnG{\Nt'}{\Nt''}\\
      \{\T\} & \text{If } \Nt \Production \T\\
      \{\varepsilon\} & \text{If } \Nt \Production \varepsilon\\
    \end{array}  \right. \label{eq:lang} \end{equation}
	The language $\Lang{\Gram}$ generated by a grammar $\Gram=(\Voc,\NtSet,\ProdRules,\Axiom)$
	is then defined as $\Lang{\Axiom}$ the language associated with the axiom $\Axiom$.
    Finally, let us denote by $\mathcal{L}_n$ the restriction of a language $\mathcal{L}$ to words of length $n$.

	\subsection{Weighted context-free grammars}
	\begin{defn}[Weighted Grammar~\cite{deniserandom}]
		A weighted grammar $\GramW$ is a 5-tuple $\GramW=(\Pond,\Voc,\NtSet,\ProdRules,\Axiom)$ where
		$(\Voc, \NtSet, \ProdRules,\Axiom)$ define a context-free grammar and
		$\Pond:\Voc\to\mathbb{R}^+$ is a weighting function that associates a real-valued weight $\Pond_\T$ to each terminal symbols $\T$.
	\end{defn}
	This notion of weight naturally extends to any mature word $w$ in a multiplicative fashion,
	i.e. such that $\Pond(w)=\prod_{i=1}^{|w|}\Pond_{w_i}$. It also extends additively on any set of words $\mathcal{L}$ through
    $\Pond(\mathcal{L})= \sum_{w\in \mathcal{L}} \Pond(w)$. One defines a
	$\Pond$-\Def{weighted probability distribution} over $\mathcal{L}$ such that
		\begin{equation} \label{eq:relativeProb} \Prob{w\;|\;\Pond, \mathcal{L}} = \frac{\displaystyle{\Pond(w)}}{\displaystyle{\sum_{w'\in\mathcal{L}}\Pond(w')}} = \frac{\displaystyle{\Pond(w)}}{\displaystyle{\Pond(\mathcal{L})}},\; \forall w\in\mathcal{L}.\end{equation}

	The random generation of words of a given length $n$ with respect to a weighted probability
	distribution has been addressed by previous works, and an algorithm in $\BigO{n\log{n}}$ after $\BigO{n^2}$ arithmetic operations was
	described~\cite{deniserandom} and implemented~\cite{PoTeDe06}.
	
  \subsection{Problem statement}

 \begin{algorithm}[t!]
\caption{Non-redundant sequential meta-algorithm for the generation of $k$ distinct words of length $n$, from a (weighted) context-free grammar $\GramW=(\Pond,\Voc,\NtSet,\ProdRules,\Axiom)$, avoiding a forbidden set of words $\Forb$.}
\label{alg:sequentialnonred}
{\bf NonRedundantSequential$(\GramW,k,n,\Forb)$}:
\begin{algorithmic}
\STATE Perform some precomputations\ldots
\STATE $\R \leftarrow \varnothing$
\WHILE {$|\R|\le k$}
	\STATE $x \leftarrow {\bf DrawNonRed}(\Axiom_n,\Pond(\Nt_n),\GramW,\Forb)$ \COMMENT{Any non-redundant algorithm}
    \STATE Update some data structure\ldots
    \STATE $(\R,\Forb) \leftarrow (\R\cup\{x\},\Forb\cup\{x\})$
\ENDWHILE
\RETURN $\R$
\end{algorithmic}
\end{algorithm}

  In the following, we consider algorithmic solutions for the non-redundant generation of a collection of words of a given length, generated by an unambiguous weighted context-free grammar. Our precise goal is to simulate efficiently a sequence of independent calls to a random generation algorithm until a set of exactly $k$ distinct words in a language $\mathcal{L}$ are obtained. 
The returned subset $\R\subseteq\mathcal{L}_n$, $|\R|=k$, can be generated in any order, and the random generation scenarios leading to an ordering $\sigma$ of $\R$ can be decomposed as:
$$ \sigma_1 \to  \sigma_1^* \to \sigma_2 \to (\sigma_1\;|\;\sigma_2)^*\to \ldots \to \sigma_{k-1} \to (\sigma_1\;|\;\cdots\;|\;\sigma_{k-1})^*\to \sigma_k.$$
The successive calls made to the weighted random generator are independent, therefore the total probability $p_\sigma$ of getting a set  $\R$ in a given order $\sigma$ is given by
\begin{align*} 
 p_\sigma&=\frac{\Pond(\sigma_{1})}{\Pond(\mathcal{L}_n)}\cdot \prod_{i=2}^k\left(\sum_{m\ge0}\left(\frac{\sum_{j=1}^{i-1}\Pond(\sigma_j)}{\Pond(\mathcal{L}_n)}\right)^m \cdot\frac{\Pond(\sigma_{i})}{\Pond(\mathcal{L}_n)}\right)\\
 &=\frac{\Pond(\sigma_{1})}{\Pond(\mathcal{L}_n)}\cdot \prod_{i=2}^k\left(\frac{1}{1-\frac{\sum_{j=1}^{i-1}\Pond(\sigma_j)}{\Pond(\mathcal{L}_n)}}
\cdot\frac{\Pond(\sigma_{i})}{\Pond(\mathcal{L}_n)}\right) = \prod_{i=1}^k\left(\frac{\Pond(\sigma_{i})}{\Pond(\mathcal{L}_n)-\sum_{j=1}^{i-1}\Pond(\sigma_j)}
\right).
 \end{align*}
Summing over every possible permutation of the elements in $\R$, one obtains 
   \begin{equation} \Prob{\R\;|\;k,  n} = 
     \displaystyle\sum_{\sigma \in \mathfrak{S}(\R)}\prod_{i=1}^{k}\frac{\Pond(\sigma_i)}{\Pond(\mathcal{L}_n)-\sum_{j=1}^{i-1}\Pond(\sigma_j)}  \label{eq:setDistribution}\end{equation}
where $\mathfrak{S}(\R)$ is the set of all permutations over the elements of $\R$.  The problem can then be restated as:\\[1em]
   \phantom{\hspace{.02\textwidth}}\fbox{\begin{minipage}{.94\textwidth}
   {\noindent\sc Weighted-Non-Redundant-Generation} ({\sc wnrg})\\[.6em]
   {\noindent\sc Input: }An unambiguous weighted grammar $\GramW$ and two positive integers $n$ and $k$.\\[.6em]
   {\noindent\sc Output: }A set of words $\R\subseteq\Lang{\Gram}_n$ of cardinality $k$  with probability $\Prob{\R\;|\;k, n}$.
   \end{minipage}}\\[1em]

   Note that the distribution described by Equation~\eqref{eq:setDistribution} naturally arises from a sequence of dependent calls
   $(r_1,\ldots,r_k)$ to weighted generators for $\mathcal{L}$, avoiding sets of words $\varnothing$, $\{r_1\}$, \ldots, $\{r_1,\ldots,r_{k-1}\}$ respectively, as implemented in Algorithm~\ref{alg:sequentialnonred}. It is therefore sufficient to address the generation of a single word $w$, while avoiding a prescribed set $\Forb$, in the weighted probability distribution $\Prob{w\;|\;\Pond, \mathcal{L}\backslash\Forb}$.

\section{Naive rejection algorithm}\label{sec:rejet}

	A \Def{naive rejection strategy} for this problem consists in drawing words at
	random in an unconstrained way, rejecting those from the forbidden set until a valid word
	is generated, as implemented in Algorithm~\ref{alg:reject}. As noted by Zimmermann~\cite{ZimPowerset95}, this approach is suitable for the uniform distribution
    of objects in general recursive specifications.
    This rejection strategy  relies on an auxiliary generator  {\bf draw($\cdots$)} of words
    from a (weighted) context-free languages, and we refer to
	previous works by Flajolet \emph{et al}~\cite{flajoletcalculus,fullboltz}, or
	Denise \emph{et al}~\cite{DeZi99} for efficient solutions for this problem. 

\begin{prop}[Correctness of a naive rejection algorithm]
  Any word returned by Algorithm~\ref{alg:reject} is drawn with respect to the weighted distribution on $\Lang{\Gram}_n\backslash \Forb$.
\end{prop}
\begin{proof}
  Let $w$ be the word returned by the algorithm, and $\Forb = \{f_i\}_{i=1}^{|\Forb|}$. Let us characterize the sequences of words generated by {\bf draw}, leading to the generation of $w$, by mean of a rational expression over an alphabet $\Forb\cup \{w\}$:
  $$\mathcal{R}_{w} = (f_1\;|\;f_2\;|\;\cdots\;|\;f_{|\Forb|})^*.w.$$
  Let $p_x = {\Pond(x)}/{\Pond(\Lang{\Gram}_n)}$ the probability of emission of any -- possibly forbidden -- word $x\in \Lang{\Gram}_n$, then the cumulated probability of the sequences of calls to {\bf draw}, leading to the generation of $w$, is such that
  \begin{align*}
    \Prob{\mathbf{w}} &= p_w+\left(\sum_{i=1}^{|\Forb|}p_{f_i}\right)\cdot p_w + \left(\sum_{i=1}^{|\Forb|}p_{f_i}\right)\cdot \left(\sum_{i=1}^{|\Forb|}p_{f_i}\right)\cdot p_w + \cdots\\
    &= \frac{p_w}{1-\sum_{i=1}^{|\Forb|}p_{f_i}} = \frac{\Pond(w)}{\Pond(\Lang{\Gram}_n) -\sum_{i=1}^{|\Forb|}\Pond(f_i)} = \frac{\Pond(w)}{\Pond(\Lang{\Gram}_n\backslash \Forb)}.
  \end{align*}
\end{proof}
\begin{algorithm}[t]
\caption{Naive rejection algorithm for generating a word of length $n$, from a (weighted) context-free grammar $\GramW$, avoiding a forbidden set of words $\Forb$.}
\label{alg:reject}
{\bf NaiveRejection$(\GramW,n,\Forb)$}:
\begin{algorithmic}
\REPEAT
	\STATE $t \leftarrow \text{\bf draw}(\GramW,n)$ \hfill \COMMENT{One may use any available generation algorithm.}
\UNTIL{$t \notin \Forb$}
\RETURN $t$
\end{algorithmic}
\end{algorithm}
\subsection{Complexity analysis: Uniform distribution}
  Let us analyze the complexity of Algorithm~\ref{alg:reject}, given $\mathcal{L}$ a context-free language, $n\in\mathbb{N}^+$ a positive integer and $\Forb\subset\mathcal{L}_n$ a set of forbidden words, assuming a uniform distribution on $\mathcal{L}_n$. 

 One first remarks that the worst-case time-complexity of the algorithm is unbounded, as nothing prevents the algorithm from repeatedly generating the same word. An average-case analysis, however, draws a more contrasted picture of the time complexity.
	\begin{thm}
		In the uniform distribution, the naive rejection implemented in Algorithm~\ref{alg:reject} leads to an average-case complexity in
		 $\mathcal{O}\left(\left(\frac{|\mathcal{L}_n|}{|\mathcal{L}_n|-|\Forb|}\right)\cdot k\log{k}\cdot  \Draw{n}\right)$, where $\Draw{n}$ is the complexity of drawing a single word.
	\end{thm}
	\begin{proof}
		In the uniform model when $\Forb=\varnothing$, the number of attempts required by the generation of the
		$i$-th word only depends on $i$ and is independent from prior events.
        Thus the expected number $X_{n,k}$ of attempts for $k$ distinct words of size $n$ is given by
		$$ \Expect{X_{n,k}} = \sum_{i=0}^{k-1}\frac{l_n}{l_n-i} = l_n(\mathcal{H}_{l_n}-\mathcal{H}_{l_n-k}) $$
		where $l_{n}\Eqdef|\mathcal{L}_n|$ is the number of words of size $n$ in the language and $\mathcal{H}_i$ the harmonic number of order $i$, as
		pointed out by Flajolet \emph{et al}~\cite{FlaGarThi92}. It follows that $\Expect{X_{n,k}}$ is trivially
		increasing with $k$, while remaining upper bounded by $k\cdot\mathcal{H}_{k}\in\Theta(k\log(k))$ when $k=l_n$
		(Coupon collector problem). Since the expected number of rejections due to a non-empty forbidden set $\Forb$ remains the same
		throughout the generation, and does not have any influence over the generated sequences, it can be
		considered independently and contributes to a factor $\frac{|\mathcal{L}_n|}{|\mathcal{L}_n|-|\Forb|}$.
	\end{proof}
	It follows that, unless the forbidden set dominates the set of words, the \emph{per-sample} complexity of the naive rejection strategy remains largely unaffected (at most a factor $\BigO{\log k}$, i.e. $\Omega(n)$ since $k\in\Omega(|\Voc|^n)$) by the cumulated cost of rejections. 

		\subsection{Complexity analysis: Weighted languages}
  Turning towards {\bf weighted context-free languages}, one shows that a rejection strategy
 may have average-case complexity which is {\bf exponential on $k$}, even in the most favorable case of an empty initial set of forbidden words.

    \begin{prop}
      The generation of $k$ distinct words, starting from an empty initial forbidden set $\Forb=\varnothing$, may require a number of calls to {\bf draw} that is exponential on $k$.
    \end{prop}
    \begin{proof}
	Consider the following grammar, generating the language denoted by the regular expression $a^*b^*$:
	\begin{align*}
		S & \Production  \PrG{a}{S}\;|\; T &
		T & \Production  \PrG{b}{T}\;|\; \varepsilon
	\end{align*}
	We adjoin a weight function $\Pond$ to this grammar, such that $\Pond(b)\Eqdef\alpha>1$ and $\Pond(a)\Eqdef 1$.
	The probability of any word $\omega_m\Eqdef a^{n-m}b^m$ in the language is
	$$\Prob{\omega_m}
		=\frac{\Pond(\omega_m)}{{\sum_{\substack{\omega\in\Lang{S}\\|\omega|=n}} \Pond(\omega)}}
		=\frac{\alpha^m}{\sum_{i=0}^{n}\alpha^i}
		=\frac{\alpha^{m+1}-\alpha^{m}}{\alpha^{n+1}-1}
		< \alpha^{m-n}.$$
	Now consider the set $\mathcal{V}_{n,k}\subset\mathcal{S}_n$ of words having less
	than $n-k$ occurrences of the symbol $b$. The probability of generating a word from
	$\mathcal{V}_{n,k}$ is then
	$$
		\Prob{\mathcal{V}_{n,k}}
		=\sum_{i= 0}^{n-k}\Prob{\omega_{n-k-i}}
		=\frac{\alpha^{n-k+1}-1}{\alpha^{n+1}-1}< \alpha^{-k}
	$$
	The expected number of generations before generating any element of $\mathcal{V}_{n,k}$
	is greater than $\alpha^k$. Since any non-redundant set of $k$ sequences issued from $\mathcal{S}_n$
	must contain at least one sequence from $\mathcal{V}_{n,k}$, then the average-case time complexity
	of a naive rejection approach is in $\Omega(n\cdot\alpha^k)$, i.e. exponential on $k$ the number of words.
    \end{proof}

	However, the above example is based on a regular language, and may not be typical of the rejection algorithm's
	behavior on general context-free languages. Indeed, it can be shown that, under a natural assumption,
	no single word can asymptotically contribute a significant portion of the distribution in simple type grammars.
	\begin{prop}
		Let $\GramW=(\Pond,\Voc,\NtSet,\Axiom,\ProdRules)$ be a weighted grammar of simple type\footnote{A grammar of
		simple type is mainly a grammar whose dependency graph is strongly-connected and whose number of words follow
		an aperiodic progression (See~\cite{FlaFusPiv07} for a more complete definition). Such a grammar can easily be
		found for the avatars of the algebraic class of combinatorial structures (Dyck words, Motzkin paths, trees of
		fixed degree,...), all of which can be interpreted as trees.}. 
   Assume that $\omega^\triangle_n$ the most probable (i.e. largest weight w.r.t. $\Pond$) word of length $n$
   has weight $\Pond(\omega^\triangle_n) \in\Theta(\alpha^n)$, for some $\alpha>0$.\\
		Then the probability of $\omega^\triangle$ decreases exponentially as $n\to\infty$:
		$$\exists\, \beta<1\text{ such that } \Prob{\omega^\triangle\;|\;\Pond}=\frac{\Pond(\omega^\triangle)}{\Pond(\Lang{\GramW}_n)} \in \Omega(\beta^n).$$
	\end{prop}
	\begin{proof}
		The Drmota-Lalley-Woods theorem~\cite{Drmota97,Lalley93,Woods93} establishes that the generating
		function of any simple type grammar has a \emph{square-root type} singularity.
		This powerful result relies on properties of the underlying system of functional equations, and therefore also holds for the coefficients of weighted generating functions~\cite{Denise2010}.
   Therefore the overall weights $W_n\Eqdef \Pond(\Lang{\GramW}_n)$ -- the coefficients of the weighted generating function -- 
   follow an expansion of the form $\frac{\kappa'\cdot \alpha'^{n}}{n\sqrt{n}}(1+\mathcal{O}(1/n))$, $\alpha',\kappa'>0$. Since $\omega^\triangle_n$ is contributing to $W_n$, then
		one has $\Pond(\omega^\triangle_n)\le\Pond(\Lang{\GramW}_n)$ and therefore $\alpha<\alpha'$. The proposition follows directly from taking $\beta:=\alpha'/\alpha.$
	\end{proof}

	Furthermore, one can easily design disconnected grammars such that, for any fixed length $n$, a subset of words $\mathcal{M}\subset \Lang{\GramW}_n$
	having maximal number of occurrences of a given symbol $\T$  has total cumulated probability $1-\alpha^n$, $0<\alpha<1$, in the weighted distribution. 
  It follows that sampling more than $|\mathcal{M}|$ words (e.g. a polynomial number of such words) can	be extremely time-consuming (typically requiring exponential-time in $n$).

    Finally, it is worth noticing that, in non-degenerate context-free languages, the weight of the least probable word $\omega^\nabla_n$ grows like
    $\Theta(\alpha^n)$, $\alpha<1$, where the exact value of $\alpha$ depends on a subtle trade-off between structural properties of the language and its weight function $\Pond$.
    In particular, $\alpha$ can become arbitrarily close to $0$, by adequately increasing the weight $\Pond(\T)$ of some terminal symbols.
    Sampling $k=|\Lang{\GramW}_n|$ words (Coupon Collector) then requires an expected $\Omega(\alpha^{-n})$ number of calls to {\bf draw}, since the waiting time of the least probable word is clearly a lower-bound for the full collection. 
    Since the number of words in a context-free language is bounded by $|\Voc|^n$ and does not depend on the weight, 
    then the \emph{average cost per generation} may grow exponentially on $n$. 
    This observation generalizes to many weighted languages, as shown by in Boisberranger~\emph{et al}~\cite{Boisberranger2012}.

\section{A step-by-step recursive algorithm}\label{sec:unifForb}
		\begin{figure}[t]
					 \resizebox{\textwidth}{!}{{\begin{tikzpicture}[inner sep=0]
  \node (s1) at (0,0) {\includegraphics[height=.15\textheight]{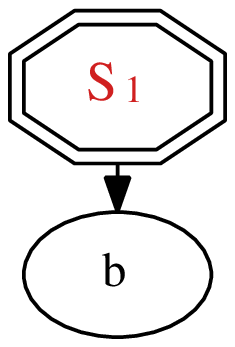}};
  \node[right=2.5em of s1.north east,anchor=north west] (s3)   {\includegraphics[height=.22\textheight]{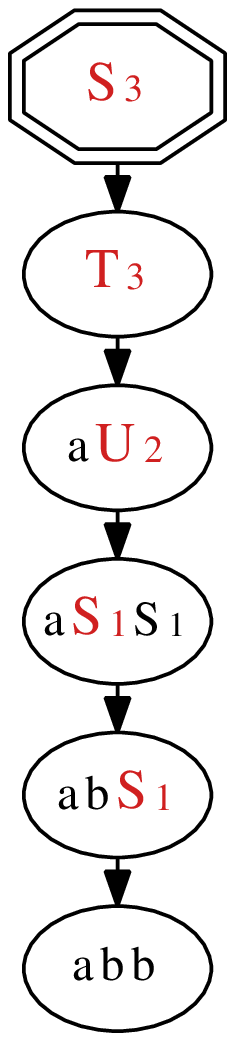}};
  \node[right=2.5em of s3.north east,anchor=north west] (s5)  {\includegraphics[height=.25\textheight]{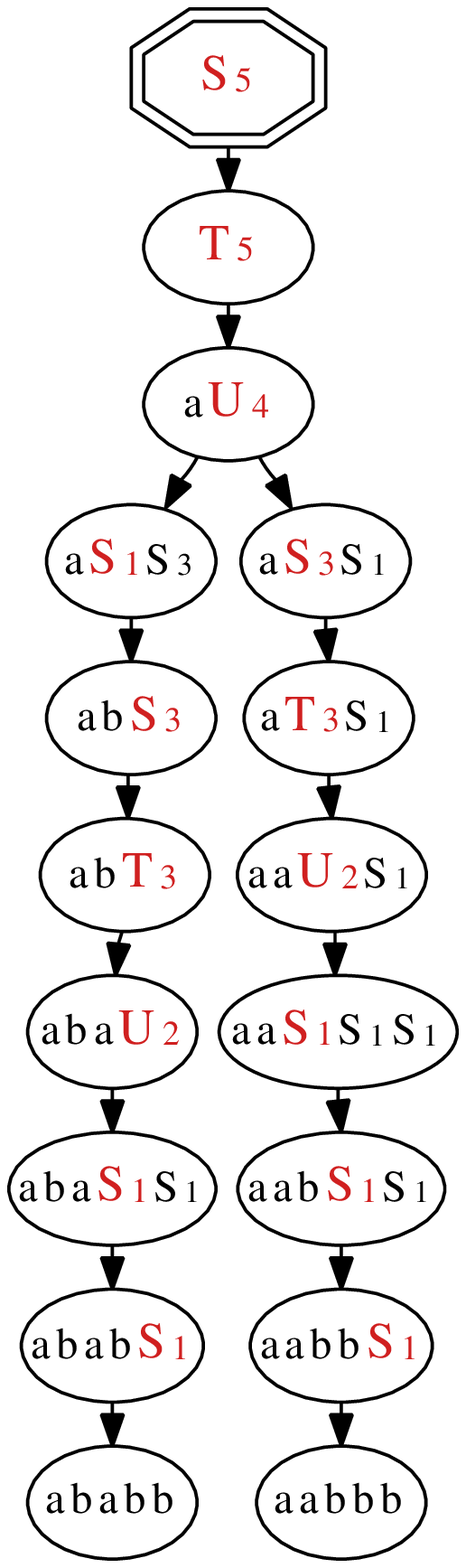}};
  \node[right=0em of s5.north east,anchor=north west] (s7)  {\includegraphics[height=.42\textheight]{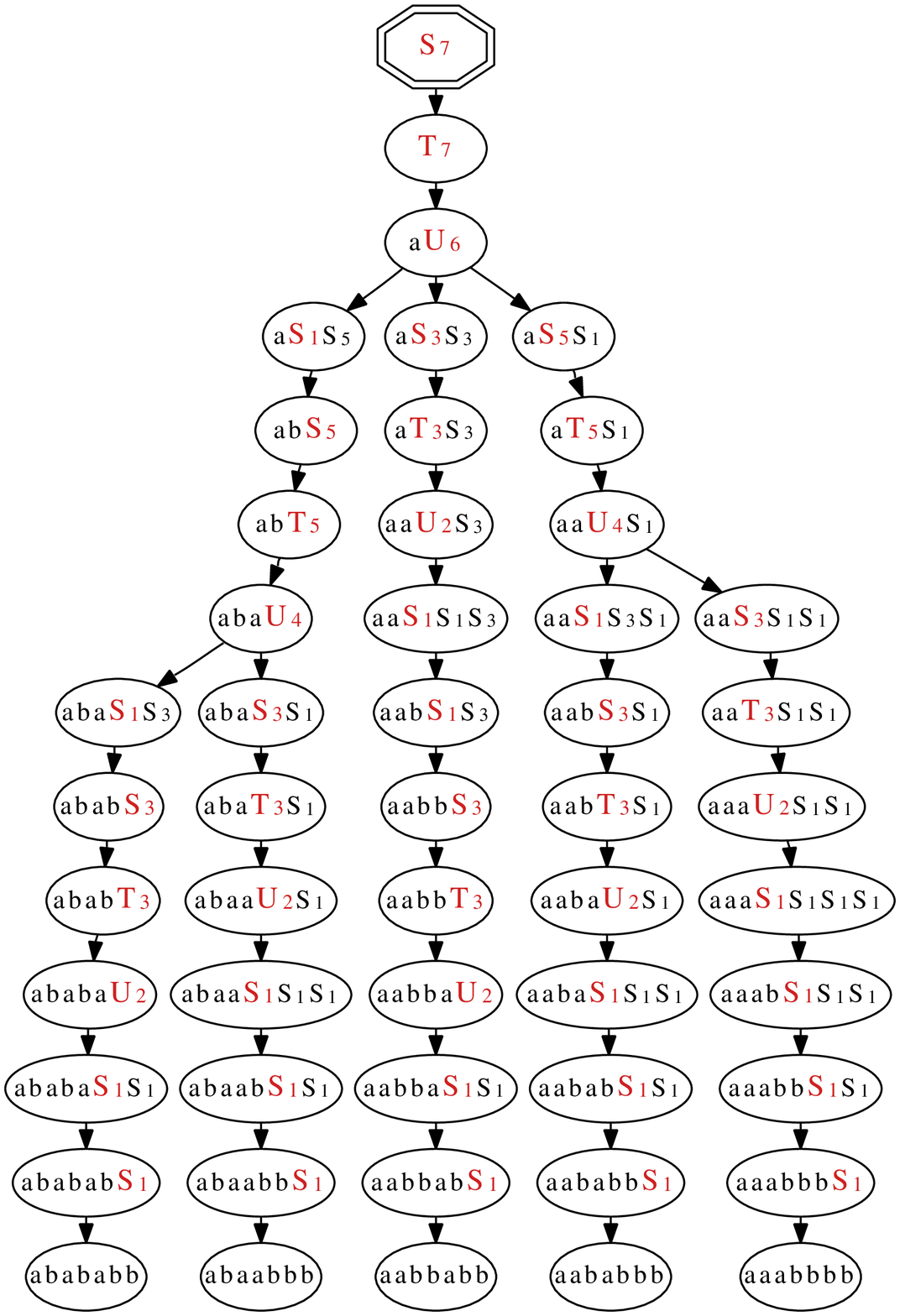}};
  \node[left=0em of s7.south west,anchor=south east] (s9) {\includegraphics[height=.20\textheight]{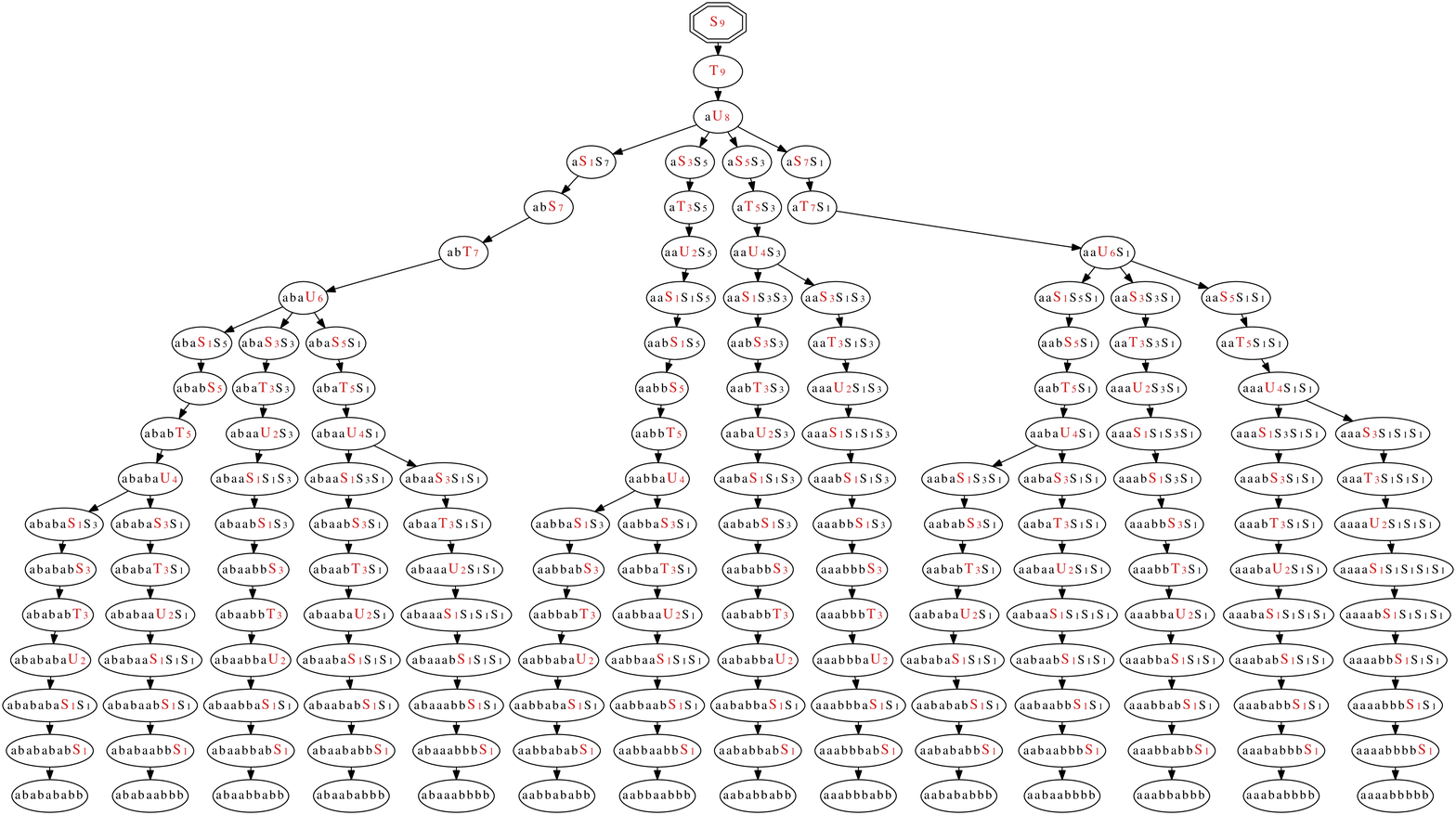}};
  \node[above=3pt of s1] {$n=1$};
  \node[above=3pt of s3] {$n=3$};
  \node[above=3pt of s5] {$n=5$};
  \node[above=3pt of s7] {$n=7$};
  \node[] at ($(s9.north)+(-2,-.5)$) {$n=9$};
\end{tikzpicture}}}

				\caption{Trees of all walks associated with prefix notations of binary trees, having length $n\in[1,9]$ and generated by the BCNF grammar $\{S\Production \UnG{T}{b},\; T\Production \PrG{a}{U},\; U\Production \PrG{S}{S}\}$,
				 under the \emph{leftmost first} derivation policy $\DerPolLR$.}
				\label{fig:MotzkinWalks}
	\end{figure}	A common approach to random generators for combinatorial objects~\cite{flajoletcalculus,deniserandom}
    consists of treating non-terminal symbols as \Def{independent generators}.
    For instance,  generating from an union-type non-terminal $\Nt\to \Nt'.\Nt''$, involves two independent calls to
    dedicated generators for $\Nt'$ and $\Nt''$, either directly (Boltzmann sampling), or after figuring
    out suitable lengths for $\Nt'$ and $\Nt''$ (Recursive method).
    Unfortunately, avoiding a predefined set of words breaks the independence assumption.

    For instance, consider an unweighted grammar $\Gram$,
    having axiom $\Nt$, and rules:
    \begin{align*}
      \Nt& \to \Nt'.\Nt'',& \Nt'&\to a\;|\;b,&&\text{and}& \Nt''&\to a\;|\;b.
    \end{align*}
    Remark that,  starting from either $N'$ or $N''$, both the recursive method and Boltzmann sampling would chose one of the rules with probability $1/2$.
    Assume now that some set $\Forb=\{aa\}$ has to
    be avoided, and that a sequential choice of derivations is adopted such that $\Nt'$ is fully derived before taking $\Nt''$ into consideration.
    In this case, the derivation $\Nt'' \to a$ must be forbidden iff $\Nt' \to a$ was chosen.
    Moreover, the probabilities assigned to the derivations of $\Nt'$ must reflect the future unavailability of some choices for $\Nt''$. 
    One possibility is to use altered probabilities such that $\{\Nt'\to_{1/3} a,\Nt'\to_{2/3} b\}$, and introduce conditional probabilities such that
     $\{\Nt''\to_{0} a, \Nt''\to_{1} b\}$ when $\Nt' \to a$, and $\{\Nt''\to_{1/2} a, \Nt''\to_{1/2} b\}$ when $\Nt' \to b$.

    The idea behind our step-by-step algorithm is to capture this dependency sequentially, by considering random generation scenarios
	as random (parse) walks. This perspective allows to determine the total contribution of all forbidden (i.e. previously encountered)
    words for each of the locally-accessible alternatives. These contributions can then be used to modify conditionally the precomputed probabilities, 
    leading to an uniform (resp. weighted) generation within $\Lang{\Gram}_n/\Forb$, while keeping the computational cost to a reasonable level.	

    \subsection{Immature words: A compact  description of fixed-length sublanguages}
		\begin{figure*}[t]
				{\centering                 \resizebox{\textwidth}{!}{
\begin{tikzpicture}[draw=black, inner sep=0]
  \node (s1) at (10pt,0) {\includegraphics[width=450pt]{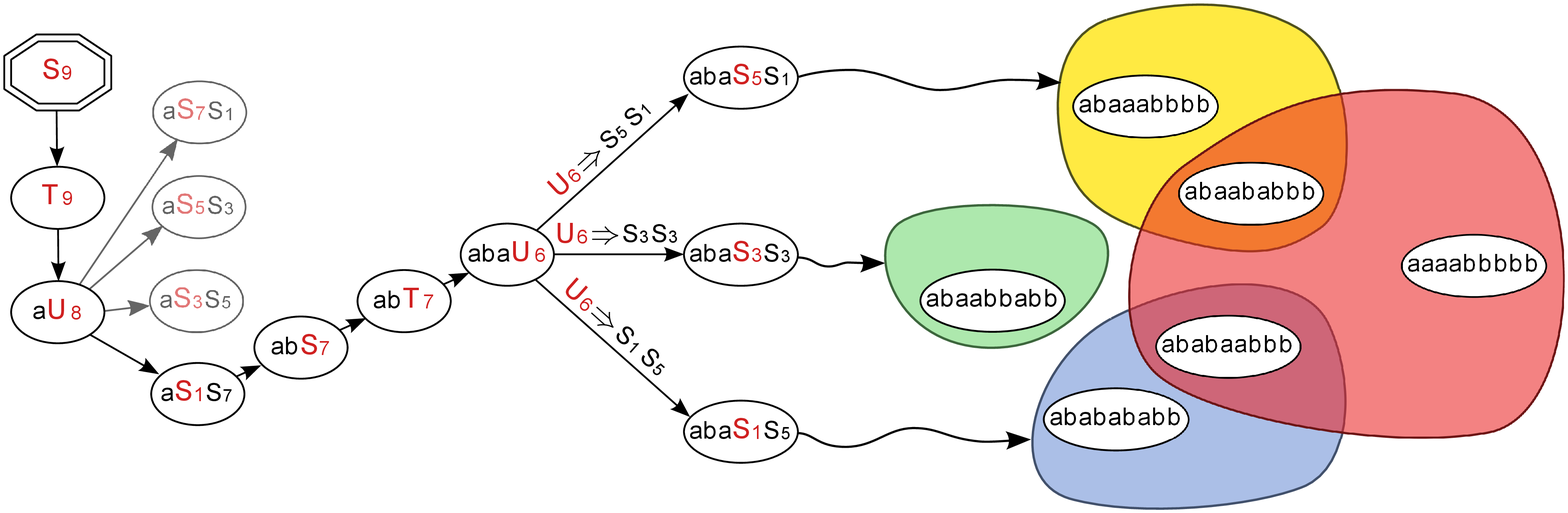}}; 
  \node at (175pt,-1pt) {\relsize{+4}$\mathcal{F}$}; 
  \node at (131pt,60pt) {$\pi($\sf abaS$_5$S$_1)$}; 
  \node at (120pt,-63pt) {$\pi($\sf abaS$_1$S$_5)$}; 
  \node at (70pt,5pt) {$\pi($\sf abaS$_3$S$_3)$}; 
  \node[] (f1)  at (160pt,30pt) {}; 
  \node[draw=orange!90!black,thick,fill=orange!30,rounded corners=2pt,inner sep=2pt] (l1) at (205pt,60pt) {$\overline{\pi}_\mathcal{F}($\sf abaS$_5$S$_1)$}; 
  \node (f2)  at (160pt,-35pt) {}; 
  \node[draw=purple!90,thick,fill=purple!20,rounded corners=2pt,inner sep=2pt] (l2) at (205pt,-65pt) {$\overline{\pi}_\mathcal{F}($\sf abaS$_1$S$_5)$}; 
  \path (l1) edge[-latex] (f1.center);
  \path (l2) edge[-latex] (f2.center);  
\node at (-3pt,68pt) { $p_1$};
\node at (-3pt,17pt) { $p_2$};
\node at (-3pt,-34pt) { $p_3$};
\node[text width=450pt,anchor=north] at (10pt,-80pt) 
  {  \begin{align*}
   p_1 &=\frac{\pi(\text{\sf abaS$_5$S$_1$})-\overline{\pi}(\text{\sf abaS$_5$S$_1$})}{\pi(\text{\sf abaU$_6$})-\overline{\pi}(\text{\sf abaU$_6$})}
  & p_2&=\frac{\pi(\text{\sf abaS$_3$S$_3$})-\overline{\pi}(\text{\sf abaS$_3$S$_3$})}{\pi(\text{\sf abaU$_6$})-\overline{\pi}(\text{\sf abaU$_6$})}
  & p_3&=\frac{\pi(\text{\sf abaS$_1$S$_5$})-\overline{\pi}(\text{\sf abaS$_1$S$_5$})}{\pi(\text{\sf abaU$_6$})-\overline{\pi}(\text{\sf abaU$_6$})}\\
    &= \frac{\pi(\text{\sf abaaabbbb})}{\pi(\text{\sf abaU$_6$})-\overline{\pi}(\text{\sf abaU$_6$})} 
  &&=\frac{\pi(\text{\sf abaabbabb})}{\pi(\text{\sf abaU$_6$})-\overline{\pi}(\text{\sf abaU$_6$})}
  &&=\frac{\pi(\text{\sf ababababb})}{\pi(\text{\sf abaU$_6$})-\overline{\pi}(\text{\sf abaU$_6$})} \\
    &\propto \pi(\mathcal{L}(\text{\sf abaS$_5$S$_1$})-\mathcal{F}) 
  &&\propto \pi(\mathcal{L}(\text{\sf abaS$_3$S$_3$})-\mathcal{F}) 
  &&\propto \pi(\mathcal{L}(\text{\sf abaS$_1$S$_5$})-\mathcal{F}) 
  \end{align*} 
};
\end{tikzpicture}  }}
				\caption{Snapshot of a \emph{step-by-step} random scenario for the language consisting of prefix notations of binary trees of
				length $6$, generated while avoiding $\Forb$. The step-by-step algorithm chooses
                one out of three possible derivations for $\sf aba{\color{red}U_6}$ using probabilities proportional to the overall weights of accessible/admissible
                words.}
				\label{fig:FullExample}
	\end{figure*}
	Let us introduce the notion of \Def{immature} words, defined as words on both the terminal and non-terminal alphabets,
    where \Def{prescribed lengths} are additionally attached to any occurrence of a symbol.
    Formally, let $\Gram=(\Voc,\NtSet,\ProdRules,\Axiom)$ be a context-free grammar, then an immature word is any word
    $$\omega\in \LangImm{\Gram} \subseteq \left((\Voc\cup \NtSet)\times \mathbb{N}^+\right)^*,$$
    where $\LangImm{\Gram}$ is the set of immature words generated from the axiom $\Axiom$.
    Such words may contain non-terminal symbols, and potentially require some further derivations before
    becoming a word on the terminal alphabet, or mature word. Intuitively, immature words correspond to intermediate states
    in a random generation scenario.

    The language associated with an immature word $\omega$ is derived from the languages of its symbols through
	\begin{equation} \Lang{\omega} = \prod_{\substack{i\in [1,|\omega|]\\s_m = \omega_i}}\Lang{s}_{m}\end{equation}
    where $\Lang{s}$ is defined as in Equation~\ref{eq:lang} with $s\in \NtSet$, and naturally extended on terminal symbols $t\in\Voc$ through $\Lang{t} = \{t\}$.
    In the following, we use the notation $\Pond(\omega)$ as a natural shorthand for $\Pond(\Lang{\omega})$ and denote by $\NtForbs{\omega}\Eqdef \Pond(\Lang{\omega}\cap \Forb)$ the total weight of all forbidden words in $\Lang{\omega}$.

	\subsection{Random generation as a random walk in language space}
	An \Def{atomic derivation}, starting from a word $\omega=\PrG{\PrG{\omega'}{\Nt}}{\omega''}\in\{\Voc\cup\NtSet\}^*$,
	is the application of a production $\Nt\Production X$ to $\omega$, that replaces $\Nt$ by the right-hand
	side $X$ of the production, yielding $\omega\Derive \omega'.X.\omega''$.
	Let us call \Def{derivation policy} a deterministic strategy that points, in an immature word, to some non-terminal 
 to be rewritten through an atomic derivation. Formally, a derivation policy is a function $\DerPol:\Lang{\Gram}\cup\LangImm{\Gram}\to\mathbb{N}\cup\{\varnothing\}$ such
	that
	$$
		\begin{array}{rrcl}
			\DerPol : & \omega\in\Lang{\Gram}& \to & \varnothing\\
							 & \omega'\in\LangImm{\Gram}& \to & i\in [1,|\omega'|] \text{ such that } \omega_i\in\NtSet .
		\end{array}
	$$

	The \Def{unambiguity} of a grammar requires that any generated word be generated by a unique sequence of derivation.
    A sequence of atomic derivations is then said to be \Def{consistent with a given derivation policy} if
	the non-terminal rewritten at each step is the one pointed by the policy.
	This notion provides a convenient
	framework for defining the \Def{unambiguity} of a grammar without explicit reference to parse trees.
	\begin{defn}[Unambiguity]
		Let $\Gram=(\Voc,\NtSet,\ProdRules,\Axiom)$ be a context-free grammar and $\DerPol$ a derivation policy acting on $\Gram$.
		The grammar $\Gram$ is said to be {\bf unambiguous} if and only if, for each $\omega\in\Voc^*$, there exists
		at most one sequence of atomic derivations that is consistent with  $\DerPol$ and produces $\omega$ from $\Axiom$.
	\end{defn}

		Any derivation leading to a mature word $\omega\in\Lang{\Gram}$ in an unambiguous grammar $\Gram$ can then be associated,
	in a one-to-one fashion, with a walk in the space of sublanguages associated with immature words, or \Def{parse walk},
	taking steps consistent with a given derivation policy $\phi$. More
	precisely, such a walk starts from the axiom $\Axiom$ and, for any intermediate immature word $X\in\LangImm{\Gram}$, the derivation policy $\phi$
    points at a position $\phi(X)$, where a non-terminal $X_k$ can be found. The parse walk can then be
	extended using one of the derivations acting on $X_k$ (See Figures~\ref{fig:MotzkinWalks} and~\ref{fig:FullExample}), until a mature word in
    $\Voc^*$ is reached.

\begin{algorithm}[t!]
\caption{Step-by-step random generation algorithm. $\GramW$ is a weighted grammar, $\omega$ an immature word,
    $\mu=\Pond(\omega)$ is the {\bf precomputed} weight of the language generated from $\omega$, and $\Forb\subset\Lang{\omega}$ is a set of forbidden words.}
\label{alg:stepbystep}
$\StepByStep(\omega,\mu,\GramW,\Forb,\DerPol):$
\begin{algorithmic}[5]
\IF {$\mu \le \NtForbs{\omega}$}
  \RETURN{{\tt Error}}
\ELSIF {$\DerPol(\omega)=\varnothing$}
  \RETURN $\omega$ \hfill\COMMENT{$\omega$ is a mature word, generation is over}
\ENDIF
  \STATE $(\omega',\Nt_m,\omega'') \leftarrow \left(\omega_{[1,\DerPol(\omega)-1]},\omega_{\DerPol(\omega)},\omega_{[\DerPol(\omega)+1,|\omega|]}\right)$
  \STATE $ r\leftarrow {\bf rand}(\mu - \NtForbs{\omega})$ \hfill \COMMENT{$r$ is random, uniformly in $[0,\Pond(\Lang{\omega}/\Forb))$}
  \IF[{\bf Union} type] {$\Nt \Production \UnG{\Nt'}{\Nt''}$}
    \STATE $\mu' \leftarrow \mu\cdot \Pond(\Nt'_m)/\Pond(\Nt_m)$
    \STATE{$r \leftarrow r - (\mu'-\NtForbs{\omega'.\Nt'_m.\omega''})$}
    \IF{$r<0$}
    \RETURN $\StepByStep(\omega'.\Nt'_m.\omega'',\mu',\GramW,\Forb)$
    \ELSE
    \RETURN $\StepByStep(\omega'.\Nt''_m.\omega'',\mu\cdot \Pond(\Nt''_m)/\Pond(\Nt_m),\GramW,\Forb)$
    \ENDIF
  \ELSIF[{\bf Product} type] {$\Nt \Production \PrG{\Nt'}{\Nt''}$}
    \FORALL[Boustrophedon order $1,n-1,2,n-2\ldots$]{$i \in [1,n-1]$}
    \STATE $\mu_i \leftarrow \mu\cdot \Pond(\Nt'_i)\cdot\Pond(\Nt''_{m-i})/\Pond(\Nt_m)$ \label{alg:stepbysteplineproduct}

    \STATE{$r \leftarrow r - (\mu_i-\NtForbs{\omega'.\Nt'_i.\Nt''_{m-i}.\omega''})$}
    \IF{$r<0$}
    \RETURN $\StepByStep(\omega'.\Nt'_i.\Nt''_{m-i}.\omega'',\mu_i,\GramW,\Forb)$
    \ENDIF
    \ENDFOR
  \ELSIF[{\bf Terminal} type]  {$\Nt \Production \T$}
      \RETURN $\StepByStep(\omega'.\T.\omega'',\mu,\GramW,\Forb,\DerPol)$
\ENDIF
\end{algorithmic}\vspace{.2em}
{Where: }{\bf rand}$(x)$: Draws a random number uniformly in $[0,x)$\\
\phantom{Where: }$\NtForbs{\omega}\Eqdef \Pond(\Lang{\omega}\cap \Forb)$: Total weight of forbidden words in $\Lang{\omega}$
\end{algorithm}

	\subsection{A step-by-step algorithm}

    Let us now describe and validate Algorithm~\ref{alg:stepbystep}, based on the recursive method introduced by Wilf~\cite{wilf77}, which uses the concepts of immature words to linearize the generation of words. More specifically, the algorithm draws a random word through a sequence of local choices (atomic derivations) using probabilities that are proportional to the cumulated weight of accessible 
non-forbidden words, as illustrated by Figure~\ref{fig:FullExample}. To grant access to such weights in reasonable time, the cumulated weights of languages generated by non-terminals are precomputed recursively~\cite{Denise2010}, and a dedicated \emph{tree-like} data structure is introduced to gain efficient access to the contribution of forbidden words.
    \begin{thm}
  Algorithm~\ref{alg:stepbystep} generates $k$ distinct words of length $n$ from a weighted grammar
  $\GramW$ in $\mathcal{O}(n\cdot |\NtSet|+k\cdot n \log n)$ arithmetic operations,
  while storing $\mathcal{O}(n\cdot |\NtSet|+k)$ numbers, and a data structure consisting of $\Theta(n\cdot k)$ nodes.
\end{thm}
\begin{proof}
  As discussed in Section~\ref{subsecComplex}, Algorithm~\ref{alg:stepbystep} generates a word in $\mathcal{O}(n\log(n))$ arithmetic operations, assuming that
  some correcting terms $\NtForbs{\omega}$ are available at runtime. In Section~\ref{subsec:forbidden tree}, a data structure is introduced that returns
  this value in $\mathcal{O}(\log(n))$ time. Namely, one has that $\mathcal{O}(n\log(n))$ times, a search in $\mathcal{O}(\log(n))$ is performed followed by
  an arithmetic operation involving large (at least polynomial on $n$, usually exponential) numbers. It follows that the cost of accessing the data structure
  is dominated by the cost of the following arithmetic operations, and the overall cost of generating $k$ words is in $\mathcal{O}(k\cdot n\log(n))$ arithmetic operations.
  After each generation, the data structure is updated in $\Theta(n)$ arithmetic operations, and the complexity is therefore dominated by the cost of the generation.

  The precomputation required by the  \StepByStep{} algorithm involves $\Theta(n\cdot |\NtSet|)$ arithmetic operations, and the storage of
  $\Theta(n\cdot |\NtSet|)$ numbers. The data structure for $\NtForbs{\omega}$ has $\Theta(n\cdot k)$ nodes and contains $\Theta(k)$ different numbers, thus
  the overall complexity.
\end{proof}

\subsection{Correctness}
	\begin{prop} \label{thm:lastTheorem}
        Assuming that $\mu = \Pond(\omega)$, Algorithm~\ref{alg:stepbystep} draws a word at random according to the $\Pond$-weighted distribution on $\Lang{\omega}\backslash\Forb$, or returns {\tt error} iff
		$\Lang{\omega}\backslash\Forb=\varnothing$.
	\end{prop}
	\begin{proof}
        Let us start with some observations to simplify the proof.
        First, since $\mu = \Pond(\omega)$, then the  variables $\mu'$ and $\mu_i$ of Algorithm~\ref{alg:stepbystep} respectively obey
        \begin{align} \mu' &= \Pond(\omega)\cdot \frac{\Pond(\Nt'_m)}{\Pond(\Nt_m)} = \frac{\Pond(\omega').\Pond(\Nt_m).\Pond(\omega'')\cdot\Pond(\Nt'_m)}{\Pond(\Nt_m)} = \Pond(\omega'.\Nt'_m.\omega'') \\
        \mu_i &= \Pond(\omega)\cdot \frac{\Pond(\Nt'_i)\cdot\Pond(\Nt''_{m-i})}{\Pond(\Nt_m)} = \Pond(\omega'.\Nt'_i.\Nt''_{m-i}.\omega''). \label{eq:mu2}
        \end{align}
        Secondly for any immature word $\omega$, one has
        $$ \Pond(\omega) - \NtForbs{\omega} = \Pond(\Lang{\omega}) - \Pond(\Lang{\omega}\cap \Forb) = \Pond(\Lang{\omega}\backslash \Forb).$$

        We now show that, provided that $\mu=\Pond(\omega)$ holds, then any word in $\Lang{\omega}$ is generated with respect to a weighted distribution on $\Lang{\omega}\backslash\Forb$.
        Let $d$ be the maximum number of recursive calls needed for the generation of a mature word from a given immature word $\omega$, then one has:

		{\bf\noindent Base:} The $d=0$ case corresponds to an already mature word $\omega$, for which the associated language is limited
		to $\{\omega\}$. In this case, $\omega$ has probability $1$ in the weighted distribution, and is indeed always generated.

		{\bf\noindent Inductive step:} Assuming that the theorem holds for $d \le n$, we investigate the probabilities of emission of
		words that require $d=n+1$ derivations. Let $\Nt_m$ be the non-terminal pointed by $\DerPol$, then:
		\begin{itemize}
			\item $\Nt \Production \UnG{\Nt'}{\Nt''}$:
				Let us first assume that the derivation $\Nt^*_m\Derive \Nt'_m$ is chosen with probability
                \begin{align*}
                \frac{\mu' - \NtForbs{\omega'.\Nt'_m.\omega''} }{\mu-\NtForbs{\omega}} =& \frac{\Pond(\omega'.\Nt'_m.\omega'') - \NtForbs{\omega'.\Nt'_m.\omega''}}{\Pond(\omega)-\NtForbs{\omega}}\\
                  =& \frac{\Pond(\Lang{\omega'.\Nt'_m.\omega''}\backslash\Forb)}{\Pond(\Lang{\omega}\backslash\Forb)}.
                \end{align*}
                The recursive call to $\StepByStep(\omega'.\Nt'_m.\omega'',\mu',\GramW,\Forb)$ indeed satisfy $\mu' = \Pond(\omega'.\Nt'_m.\omega'')$,
                and subsequently generates a mature word $x$ using at most $n$ recursive calls.
                The induction hypothesis holds, and the emission probability of $x\in \Lang{\omega'.\Nt'_m.\omega''}\backslash \Forb$ is therefore given by
				${\Pond(x)}/{\Pond(\Lang{\omega'.\Nt'_m.\omega''}\backslash\Forb)}$.
				The overall probability of issuing $x$ starting from $\omega$ is then
				\begin{align*}
                    \frac{\Pond(\Lang{\omega'.\Nt'_m.\omega''}\backslash\Forb)}{\Pond(\Lang{\omega}\backslash\Forb)} . \frac{\Pond(x)}{\Pond(\Lang{\omega'.\Nt'_m.\omega''}\backslash\Forb)} = \frac{\Pond(x)}{\Pond(\Lang{\omega}\backslash\Forb)}
				\end{align*}
                in which one recognizes the weighted distribution on $\Lang{\omega}\backslash\Forb$, and the argument applies symmetrically  to $\Nt''_m$.
			\item $\Nt \Production \PrG{\Nt'}{\Nt''}$:
				A repartition $\Nt_m \Derive \PrG{\Nt'_i}{\Nt''_{m-i}}, i\in[1,m-1]$ is chosen with probability
                \begin{align*}
                \frac{\mu' - \NtForbs{\omega'.\Nt'_i.\Nt''_{m-i}.\omega''} }{\mu-\NtForbs{\omega}} =& \frac{\Pond(\omega'.\Nt'_i.\Nt''_{m-i}.\omega'') - \NtForbs{\omega'.\Nt'_i.\Nt''_{m-i}.\omega''}}{\Pond(\omega)-\NtForbs{\omega}}\\
                  =& \frac{\Pond(\Lang{\omega'.\Nt'_i.\Nt''_{m-i}.\omega''}\backslash\Forb)}{\Pond(\Lang{\omega}\backslash\Forb)}.
                \end{align*}
                A recursive call is then made on an immature word $\omega'.\Nt'_i.\Nt''_{m-i}.\omega''$, using weight $\mu_i$.
                As established in Equation~\ref{eq:mu2}, one has $\mu_i=\Pond(\omega'.\Nt'_i.\Nt''_{m-i}.\omega'')$, therefore the induction hypothesis
                applies, and any word $x\in\Lang{\omega'.\Nt'_i.\Nt''_{m-i}.\omega''}$ is generated by the recursive call with
                probability
                $$  \frac{\Pond(x)}{\Pond(\Lang{\omega'.\Nt'_i.\Nt''_{m-i}.\omega''}\backslash\Forb)}.$$
                The emission probability of $x\in\Lang{\omega'.\Nt'_i.\Nt''_{m-i}.\omega''}$ from $\omega$ is then given by
                $$ \frac{\Pond(\Lang{\omega'.\Nt'_i.\Nt''_{m-i}.\omega''}\backslash\Forb)}{\Pond(\Lang{\omega}\backslash\Forb)}.\frac{\Pond(x)}{\Pond(\Lang{\omega'.\Nt'_i.\Nt''_{m-i}.\omega''}\backslash\Forb)} = \frac{\Pond(x)}{\Pond(\Lang{\omega}\backslash\Forb)}.$$
			\item $\Nt \Production \T$: The emission probability for any word $x$ emitted from $\omega$ equals that of the word
			issued from $\omega'.\T.\omega''$. It is then given by $\frac{\Pond(x)}{\Pond(\Lang{\omega'.\T.\omega''}\backslash\Forb)}=\frac{\Pond(x)}{\Pond(\Lang{\omega}\backslash\Forb)}$ according to the
			induction hypothesis, which applies since $\Pond(\omega'.\T.\omega'') = \Pond(\omega'.\Nt.\omega'')$.
		\end{itemize}
	\end{proof}
\subsection{Complexities and data structures} \label{subsecComplex}
	The overall complexity of Algorithm~\ref{alg:stepbystep} depends critically on efficient algorithms and data structures for:
	\begin{enumerate}
		\item Accessing the weights of languages associated with non-terminals.
		\item Computing the total weight $\NtForbs{\omega}\Eqdef\Pond(\Lang{\omega}\cap \Forb)$ of all forbidden words accessible from an immature word $\omega$.
		\item Investigating the \emph{partitions} $\Nt^*_m \Derive \PrG{\Nt'_i}{\Nt''_{m-i}}$ for \emph{product rules}.\label{alg:Boustrophedon}
		\item Handling large numbers.\label{alg:Floats}
	\end{enumerate}
	\subsubsection{Weights of non-terminals}\label{sec:weights}
		As is usual within the recursive approach~\cite{deniserandom}, the total weights $\Pond(\Nt_i)$ of languages generated from each non-terminal $\Nt$ must be readily available during the generation at generation time. A precomputation of these numbers can be performed in $\Theta(n)$ arithmetic operations, thanks to
		the algebraic, therefore holonomic, nature of the weighted counting generating functions. Indeed, the coefficients of an holonomic
		generating function obey a linear recurrence with polynomial coefficients in $n$. Such a recurrence can be algorithmically determined from the system 
   of functional equations induced by the context-free grammar (e.g. using the {\tt Maple} package {\tt GFun}~\cite{SalZim94}). 
	\subsubsection{A data structure for forbidden words}\label{subsec:forbidden tree}
		\begin{figure*}[t]
				{\centering
					 {\resizebox{\textwidth}{!}{\input{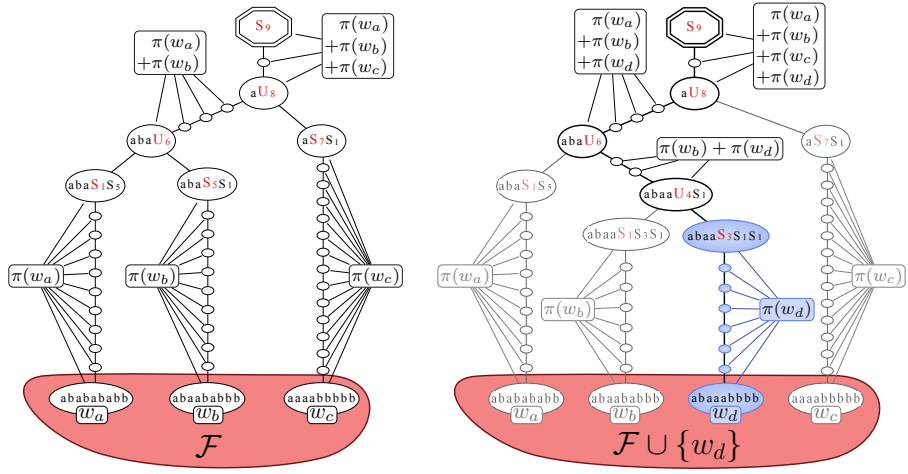}}}\\}

				\caption{Illustration of the update operation for the weighted tree for forbidden walks, for our running example. Initial tree (Left): Each node is associated with an immature word $\omega$ and its overall weight of forbidden words $\NtForbs{\omega}$ (some unary nodes are contracted for the sake of readability). During the execution of Algorithm~\ref{alg:stepbystep}, the tree is traversed to grant efficient access to $\NtForbs{\omega}$. Updated tree (Right): After generating a new (mature) word $w_d:={\sf abaaabbbb}$, the proper suffix of the parse walk is added to the tree (Blue nodes), associated with the additional weight $\Pond{w_d}$, which must then be propagated back to the root (bold branch), using at most $\Theta(n)$ arithmetic operations.}
				\label{fig:PrefixTree}
	\end{figure*}
        \begin{prop}
          The total weight $\NtForbs{\omega}$ of all forbidden words generated from an immature word $\omega$ can be accessed by Algorithm~\ref{alg:stepbystep} in $\mathcal{O}(\log(n))$ time, 
          at the cost of an update operation in $\Theta(n)$ arithmetic operations, while storing $\Theta(|\Forb|)$ additional numbers.
        \end{prop}
        \begin{proof}
					Let us first remark that, in any BCNF grammar,  any parse walk $p_n$ that produces a mature word of length $n$, 
involves $\Theta(n)$ derivations (i.e. has length in $\Theta(n)$). To that purpose, let us discuss the number of occurrences of each type of rules in $p_n$, by reasoning on the associated parse tree. First, let us observe that each letter in the mature word can be bijectively associated with the application of a terminal rule, thus $p_n$ contains exactly $n$ applications of terminal rules. Then, product rules induce a binary structure in the parse tree, whose leaves correspond to the $n$ terminal letters. Therefore, $p_n$ contains exactly $n-1$ applications of product rules. Finally, sequences of union-type rules can be found before any occurrence of a product or terminal rule. However, it should be noted that the unambiguity of the grammar forbids derivations of the form $\Nt \Rightarrow^{*} \Nt$. The length of any union-type derivations sequence therefore cannot exceed $|\NtSet|+1$. Treating $|\NtSet|$ as a constant, the total number of occurrences of union rules is then in $\mathcal{O}(n)$, and we conclude that the total number of derivations involved in $p_n$ is indeed $\Theta(n)$.

        Assume now that the parse walks of the elements of $\Forb$ are available as a set $\mathcal{T}$  of sequences of immature words.  We introduce a data structure, the \Def{weighted tree of forbidden walks}, a decorated prefix-tree whose nodes are in bijection with the set of immature words in $\mathcal{T}$, and such that the overall weight $\NtForbs{\omega}\Eqdef\Pond(\Lang{\omega}\cap \Forb)$ is attached to each node labeled $\omega$.

        The idea is to descend into the tree during the execution of Algorithm~\ref{alg:stepbystep}, simply fetching the precomputed contributions $\NtForbs{\omega}$ of forbidden words, that are attached to local nodes. Implementation-wise, an argument $g$ is added to Algorithm~\ref{alg:stepbystep} (omitted in the pseudocode for the sake of readability),  holding the node associated with $\omega$ if any, or $\varnothing$ otherwise. One then gets access in $O(1)$ operations to
        the forbidden weight $\NtForbs{\omega}$ of $\omega$, and in $O(\log(n))$ to that of its children nodes $\NtForbs{\omega'.\Nt'.\omega''}$,
        $\NtForbs{\omega'.\Nt''.\Nt''_{m-i}.\omega''}$, or $\NtForbs{\omega'.\Nt'_i.\Nt''_{m-i}.\omega''}$, e.g. using AVL trees~\cite{Adelson-Velskii1962} to store the children of a node.
        Once an atomic derivation $\omega\Derive \omega'$ is chosen at random, the suitable child $g'$ of $g$ (or $\varnothing$, if no word from $\Forb$
        can be computed from $\omega$), is fed to the recursive call.

		A \Def{tree update} operation must then be performed, as illustrated by
        Figure~\ref{fig:PrefixTree}:
        \begin{itemize}
        \item First, a \emph{top-down} stage descends into the tree, ensuring efficient access to $\NtForbs{\omega}$,
        until a new mature word $w_d$ is generated. Absent nodes are then added, corresponding to the proper suffix of the parse walk (Blue nodes). \\At each step, one needs to test the presence/absence of a given immature word within the children of the current node. Since the degree of a node is bounded by $\Theta(n)$, then this operation can be performed in $\Theta(\log(n))$ time, using a dedicated AVL tree to store the children of a node. The total time complexity of a single top-down descent is therefore in $\Theta(n\log(n))$ basic instructions.
        \item Then, a unique new weight $\NtForbs{w_d}$ is created and attached to the nodes in the proper suffix of the parse walk (Blue nodes).
        A \emph{bottom-up} stage propagates the weight of the generated (mature) word
        to his ancestors, all the way up to the root. The weights associated with branching nodes along the path are incremented by $\NtForbs{w_d}$. Since $\Theta(n)$ nodes can be found from the leaf
        to the initial immature word $\Axiom_n$, then the complexity of this stage is at most in $\Theta(n)$ arithmetic operations.
        \end{itemize}

Note that the immature words used to label each node do not require an explicit encoding (which may otherwise result in a $\Theta(n^2)$ time complexity). Indeed, the immature words found on consecutive nodes may only differ on at most two positions, owing to the binary nature of products. Therefore, one may only store the difference between consecutive immature words, leading to a space complexity in $\Theta(|\Forb|\cdot n)$ bits. 
By the same token, the memory requirement can be limited to $2\cdot|\Forb|$ large numbers, by observing that a unary node and its unique successor have same value for $\NtForbs{\cdot}$, and that the memory representation this number can be shared.
        \end{proof}

	\subsubsection{Boustrophedon order for product non-terminals}\label{sec:bous}
	For product-type non-terminal rules, one may possibly have to investigate $\Theta(n)$ possible candidate partitions of the length, leading to a worst-case complexity in $\Theta(n^2)$ arithmetic operations.
  Therefore, we use a Boustrophedon order~\cite{flajoletcalculus} $(1,n-1,2,n-2,\ldots)$ to investigate possible decompositions $\Nt_m \Derive \Nt'_i.\Nt''_{m-i}$. As previously shown~\cite{flajoletcalculus}, this simple device reduces the total number of execution of the body of the innermost loop (Algorithm~\ref{alg:stepbystep}, line~\ref{alg:stepbysteplineproduct}) to
        $\mathcal{O}(n\log(n))$ in the worst case scenario.

	\subsubsection{Arbitrary precision arithmetics}

		Although efficient algebraic generators exist even for some classes of transcendent probabilities~\cite{Flajolet2011}, it is reasonable, for
        all practical purpose, to assume that weights are provided as
        floating point numbers of bounded (yet arbitrarily large) precision. 
  Since the language is context-free, the numbers involved in the precomputations of $\Nt_i$
        and in the tree of forbidden words scale like $\mathcal{O}(\alpha^n)$ for some explicit $\alpha$.
        It follows that operations performed on such numbers may take time $\mathcal{O}(n\log(n)\log\log(n))$~\cite{Van02}, while the space occupied by their encoding grows like $\mathcal{O}(n)$.

\section{Non-redundant unranking algorithm}
\label{sec:unrank}
As an alternative approach, let us propose a weighted {\it unranking algorithm}, which consists in two distinct parts:
\begin{itemize}
\item An unranking algorithm for generating words from a weighted context free grammar, presented in Section~\ref{UnrankingSec}
\item An algorithm that samples random numbers uniformly within a \emph{gapped} union of intervals,
to be used in the unranking algorithm to ensure non-redundant generation, presented in Section~\ref{sec:modRand}.
\end{itemize}
Our main result is summarized by the following theorem.
\begin{thm}
  Using an unranking approach, $k$ distinct words of length $n$ can be generated from a weighted grammar
  $\GramW$ in $\mathcal{O}(n\cdot |\NtSet|+k\cdot n \log n)$ arithmetic operations,
  while storing $\mathcal{O}(n\cdot |\NtSet|+k)$ large numbers.
\end{thm}
\begin{proof}
  In Section~\ref{sec:unranking}, we introduce Algorithm~\ref{alg:unrank}, a general unranking procedure which transforms, in $\mathcal{O}(n\log(n))$ arithmetic operations,
  any random number drawn uniformly in the interval $[0,\Pond(\Lang{\GramW}_n)[$ into a random word in $\Lang{\GramW}_n$ with respect to a weighted distribution.
  Furthermore, Section~\ref{sec:modRand} introduces a dedicated data structure, coupled with Algorithm~\ref{modRandom} which draws numbers in the subset
  $[0,\Pond(\Lang{\GramW}_n)[$ while avoiding contributions of  forbidden words, and uses $\mathcal{O}(k\log(k))$ arithmetic operations.

  The precomputation required by the Algorithm~\ref{alg:unrank} involves $\Theta(n\cdot |\NtSet|)$ arithmetic operations, and a storage of $\Theta(n\cdot |\NtSet|)$ numbers. Maintaining the data structure used by Algorithm~\ref{modRandom} requires the storage of  $\Theta(k)$ numbers.
\end{proof}

\subsection{Weighted Unranking algorithm} \label{UnrankingSec}

Unranking algorithms, formalized by Wilf~\cite{wilf77}, usually take as input a \Def{rank} in the interval $[0,|\La|)$, for $|\La|$ the number of
words in a language, and output a word from the language that is uniquely associated to this rank according to some predefined ordering. It follows that calling an unranking procedure, starting from a uniformly-generated rank, immediately gives a uniformly generated random object.

Generic unranking algorithms have been proposed for the uniform generation of words from a context-free language~\cite{Martinez00ageneric}.
Through grammar transformations aiming at the introduction a controlled ambiguity, Weinberg and Nebel~\cite{WeinbergNebel} extended their
construct to special cases of non-uniform generation. For the sake of self-completeness, we reformulate, and mildly generalize, the above algorithms. 

\subsubsection{Statement of the problem}
For a given length $n$, let us assume a total ordering on the words in $\Lang{S}_n$, and denote by $w_1, \dots, w_{|\Lang{S}|}$ the ordered list of words in $\Lang{S}_n$.
One can then split the interval $[0, \pi(\Lang{S})[$ into $|\Lang{S}|$ pieces of width $\pi(w_1), \Pond(w_2), \dots, \Pond(w_{|\Lang{S}|})$ respectively, each piece being associated to a particular word.  Denoting the $j$-th interval by $I_j$ , one has
$$I_j=\left[\sum_{k=1}^{j-1} \pi(w_k), \sum_{k=1}^{j} \pi(w_k)\right[.$$

The goal of our generalized unranking is to take as input a number $r \in [0, \Pond(\Lang{S})[$, to figure out the interval $I_j=[L_j,R_j[$ such that $L_j\le r < R_j$, and to return the corresponding word $w_k$.  Upon starting the unranking procedure from a uniformly generated random real number in $[0,\Pond(\Lang{S})[$, this word is to be selected with probability proportional to the width of its interval, i.e. its weight. It follows that the whole procedure constitutes a random generation algorithm for the weighted probability distribution presented in Equation~\ref{eq:relativeProb}.

\subsubsection{Total ordering for words of length $n$}
For each non-terminal $\Nt\in\NtSet$, let us introduce a dedicated order relation $\cdot\OrderRel_{\Nt}\cdot$, defining an arbitrary notion of precedence on $\Lang{\Nt}_{m\le n}$ the  set of words of length $m$ generated from $\Nt$.
For the sake of simplicity, let us write $\mathcal{A} \OrderRel_{\Nt} \mathcal{B}$ as a shorthand for $a\OrderRel_{\Nt} b, \forall (a,b)\in\mathcal{A}\times\mathcal{B}$.
The order relation  $\cdot\OrderRel_{\Nt}\cdot$ is defined 
by $w\OrderRel_{\Nt} w, \forall w\in\Lang{\Nt}_{m\le n}$, and recursively defined by:

\begin{itemize}
 \item {\bf Union type} $\Nt \to \Nt'\;|\;\Nt''$. Then,  $\forall m\le n$, one has:
\begin{itemize}
\item  $\Lang{\Nt'_{m}}\OrderRel_{\Nt} \Lang{\Nt''_{m}}$;
   \item $\forall w_1,w_2\in\Lang{\Nt'_{m}}$ (resp. $\Lang{\Nt''_{m}}$), $w_1 \OrderRel_{\Nt} w_2$ iff $w_1 \OrderRel_{\Nt'} w_2$ (resp. $w_1 \OrderRel_{\Nt''} w_2$).
\end{itemize}
 \item {\bf Product type} $\Nt \to \Nt'.\Nt''$. Then,  $\forall m\le n$, $\forall j,j' \in [1,m-1]$, one has:
   \begin{itemize}
     \item If $j<j'$, then $\Lang{\Nt'_{j}.\Nt''_{m-j}}\OrderRel_{\Nt} \Lang{\Nt'_{j'}.\Nt''_{m-j'}}$;
     \item If $j=j'$ then $\forall (u,v), (u',v') \in \Lang{\Nt'_{j}}\times \Lang{\Nt''_{m-j}}$:
  \begin{itemize}
     \item If $u\OrderRel_{\Nt'} u'$, then $u.v \OrderRel_{\Nt} u'.v'$;
     \item If $u= u'$, then $u.v \OrderRel_{\Nt} u'.v'$ iff $v\OrderRel_{\Nt''} v'$.
    \end{itemize}
   \end{itemize}
 \item {\bf Terminal type} $\Nt \to\T$: $\Lang{\Nt_n}=\{t\}$, and one has $t\OrderRel_{\Nt} t$.
\end{itemize}
Let us then denote by $\cdot\OrderRel_{r}\cdot:=\cdot\OrderRel_{\Axiom}\cdot$ the order induced on the language generated by the axiom $\Axiom$ of the grammar. It is easily verified that $\cdot\OrderRel_{r}\cdot$ constitutes a total order over $\Lang{\Axiom_n}$.

\begin{figure}[t]
{\centering{\resizebox{\textwidth}{!}{
\begin{tikzpicture}   
\newcommand{\MyDim}{150pt}
\newcommand{\SmallSquareDim}{70pt} 
\newcommand{\MidSquareDim}{70pt} 
\newcommand{\BigSquareHeight}{170pt} 
\newcommand{\BigSquareWidth}{200pt} 
\newcommand{\PipeDownL}{1pt} 
\newcommand{\PipeDownR}{10pt} 
\newcommand{\PipeHeight}{10pt} 
\newcommand{\PipeWidth}{6pt} 
\newcommand{\IntMat}{40pt} 
\newcommand{\lthick}{1pt} 
\newcommand{\PipePath}[2]{
node (tmp1) at ($#1$) {}  
node (tmp2) at ($#2$) {}  
node (tmp3) at (tmp2|-tmp1) {}  
node (tmp4) at ($(tmp3)+(0,\PipeHeight)$) {} 
node (tmp11) at  ($(tmp4.center)+(\PipeWidth,0)$) {} 
node (tmp12) at ($(tmp4.center)+(0,-\PipeWidth)$) {}
($#1+(0,-\PipeDownR)$) 
-- node[pos=0] (end1) {} node[pos=1] (end2) {}  ++(\PipeWidth,0em) 
-- node[pos=1] (tmp5) {}  ++(0,\PipeDownR+\PipeHeight-\PipeWidth) 
to[in=0,out=90] ++(-\PipeWidth,\PipeWidth) 
to node[pos=1] {} (tmp11.center) 
to[in=90,out=180] node[pos=1] {} (tmp12.center)
-- node[pos=1] {} ($#2+(0,-\PipeDownL)$)
-- node[pos=1] {} ++(\PipeWidth,0em) 
|- node[pos=1] {} ($#1+(0,\PipeHeight-\PipeWidth)$) 
-- cycle;}

\newcommand{\Pipe}[2]{\path[pipe,left color=bleu1] \PipePath{#1}{#2}}
\newcommand{\FullPipe}[2]{\path[pipe,left color=bleu1,right color=bleu1!70] \PipePath{#1}{#2}}
\newcommand{\EmptyPipe}[2]{\path[pipe,fill=bleu1!5] \PipePath{#1}{#2}}
\colorlet{bleu1}{gray!65!blue!80!white}
\colorlet{bleu2}{bleu1!30!white}
\colorlet{pipe1}{bleu1!70!black}
\colorlet{square1}{bleu1!70!black}
\colorlet{unsure1}{bleu1}

\tikzstyle{squares}=[rectangle, draw=square1, left color=bleu1, right color=bleu2,minimum height=\SmallSquareDim, minimum width=\SmallSquareDim,line width=\lthick]
\tikzstyle{basicPoint}=[fill=none,draw=none, inner sep=0]
\tikzstyle{unsure}=[draw=square1,dashed,line width=\lthick]
\tikzstyle{vunsure}=[draw=unsure1,line width=\lthick,dashed]
\tikzstyle{sure}=[draw=square1,line width=\lthick]
\tikzstyle{hsure}=[draw=unsure1,line width=\lthick]
\tikzstyle{lbl}=[|-|,draw=square1,line width=1pt]
\tikzstyle{lbltxt}=[inner sep=2pt]
\tikzstyle{lblb}=[stealth-stealth,line width=1pt, draw=square1]
\tikzstyle{lblx}=[lbl,latex-, draw=square1]
\tikzstyle{pipe}=[line width=\lthick, draw=pipe1]

\node[squares] (sq1) {};
\node[squares,right=\IntMat of sq1.south east,anchor = south west,minimum height=1.3*\SmallSquareDim, minimum width=\MidSquareDim] (sq2) {};
\node[fill=none,draw=none,right=\IntMat of sq2.south east,anchor = south west,minimum height=\BigSquareHeight, minimum width=\BigSquareWidth] (sq3) {};

\node[basicPoint] (a0) at (sq3.north west){};
\node[basicPoint] (a1) at (sq3.35){};
\node[basicPoint] (a2) at (sq3.20){};
\node[basicPoint] (a3) at (sq3.-5){};
\node[basicPoint] (a4) at (sq3.-18){};
\node[basicPoint] (a5) at (sq3.-35){};
\node[basicPoint] (a6) at (sq3.south east){};

\node[basicPoint] (b0) at (sq3.north west){};
\node[basicPoint] (b1) at (sq3.128){}; 
\node[basicPoint] (b2) at (sq3.116){};
\node[basicPoint] (b3) at (sq3.97){};
\node[basicPoint] (b4) at (sq3.67){};
\node[basicPoint] (b5) at (sq3.52){};
\node[basicPoint] (b6) at (sq3.north east){};

\foreach \i in {0,...,6}
{
\node[basicPoint] (c\i) at (sq3.north west |- a\i){};
}

\foreach \i in {0,...,6}
{
\node[basicPoint] (d\i) at (sq3.south west -| b\i){};
}

\foreach \j in {0,...,6}
{
\foreach \i in {0,...,6}
{
  \node[draw=none,inner sep=0] (s-\i-\j) at (b\i.center |- a\j.center) {};
}
}

\foreach \i in {0,...,6}
{
  \path (s-\i-0.center) edge[unsure] (s-\i-6.center);
  \path (s-0-\i.center) edge[vunsure] (s-6-\i.center);
}

  \path (s-0-0.center) edge[unsure,draw=white] (s-6-0.center);
  \path (s-0-6.center) edge[unsure,draw=white,solid] (s-6-6.center);
  \path (s-0-6.center) edge[unsure] (s-6-6.center);
  \path (s-0-0.center) edge[unsure,draw=white,solid] (s-6-0.center);
  \path (s-0-0.center) edge[unsure] (s-6-0.center);

\foreach \i in {0,1,3,5}
{
  \pgfmathtruncatemacro\k{\i+1}
\foreach \j in {0,2,4,5}
{
  \pgfmathtruncatemacro\l{\j+1}
  \path (s-\i-\j.center) edge[hsure] (s-\k-\j.center);
  \path (s-\i-\j.center) edge[sure] (s-\i-\l.center);
  \path (s-\i-\l.center) edge[hsure] (s-\k-\l.center);
  \path (s-\k-\j.center) edge[sure] (s-\k-\l.center);
}
  \path (s-\i-6.center) edge[sure] (s-\k-6.center);
  \path (s-\i-0.center) edge[sure] (s-\k-0.center);
}

\FullPipe{(sq2.north west)+(+5pt,0)}{(sq1.north east)+(-10pt,0)} 
\FullPipe{(sq3.north west)+(+5pt,0)}{(sq2.north east)+(-10pt,0)} 
\FullPipe{(sq3.north west)+(+5pt,0)}{(sq2.north east)+(-10pt,0)} 
\FullPipe{(s-1-0.center)+(+5pt,0)}{(s-1-0)+(-10pt,0)} 
\FullPipe{(s-2-0)+(+5pt,0)}{(s-2-0)+(-10pt,0)} 
\EmptyPipe{(s-4-0)+(+5pt,0)}{(s-4-0)+(-10pt,0)} 
\EmptyPipe{(s-5-0)+(+5pt,0)}{(s-5-0)+(-10pt,0)} 
\Pipe{(s-3-0)+(+5pt,0)}{(s-3-0)+(-10pt,0)} 
\begin{pgfonlayer}{background}

  \path[left color=bleu1,right color=bleu2,decoration={snake,amplitude=.8pt},draw=black] decorate{($(end1)$) -- ($(s-3-3)+(15pt,0)$)} -- ($(s-4-3)+(-15pt,0)$) decorate{-- (end2.center)} -- cycle;

\foreach \i in {0,1,2}
{
\foreach \j in {0,...,5}
{
  \pgfmathtruncatemacro\k{\i+1}
  \pgfmathtruncatemacro\l{\j+1}
  \fill[left color=bleu1,right color=bleu2] (s-\i-\j.center) -- (s-\k-\j.center) -- (s-\k-\l.center) -- (s-\i-\l.center) -- cycle;
}
}

\foreach \j in {3,...,5}
{
  \pgfmathtruncatemacro\i{3}
  \pgfmathtruncatemacro\k{\i+1}
  \pgfmathtruncatemacro\l{\j+1}
  \fill[left color=bleu1,right color=bleu2] (s-\i-\j.center) -- (s-\k-\j.center) -- (s-\k-\l.center) -- (s-\i-\l.center) -- cycle;
}

  \path[left color=bleu1,right color=bleu2,decoration={snake,amplitude=.8pt},draw=bleu1!70!black] ($(s-3-2.center)+(0,-1)$) decorate{-- ($(s-4-2.center)+(0,-1)$)} -- ($(s-4-3.center)$) -- (s-3-3.center) -- cycle;

\end{pgfonlayer}{background}

 \newcommand{\ArrDist}{7pt}
  \draw[lbl] ($(sq1.south west)+(0,-\ArrDist)$) -- node[lbltxt,pos=.5,below] {\relsize{+2}$\pi(N'_{i-2})$} ($(sq1.south east)+(0,-\ArrDist)$);
  \draw[lbl] ($(sq1.north west)+(-\ArrDist,0)$) -- node[lbltxt,left,pos=.5] {\rotatebox{90}{\relsize{+2}$\pi(N''_{n-i+2})$}}   ($(sq1.south west)+(-\ArrDist,0)$);
  \draw[lbl] ($(sq2.south west)+(0,-\ArrDist)$) -- node[lbltxt,pos=.5,below] {\relsize{+2}$\pi(N'_{i-1})$} ($(sq2.south east)+(0,-\ArrDist)$);
  \draw[lbl] ($(sq2.north west)+(-\ArrDist,0)$) -- node[lbltxt,pos=.5,left] {\rotatebox{90}{\relsize{+2}$\pi(N''_{n-i+1})$}} ($(sq2.south west)+(-\ArrDist,0)$);
  \draw[lbl] ($(sq3.north west)+(0,3*\ArrDist)$) -- node[lbltxt,pos=.5,above] {\relsize{+2}$\pi(N'_{i})$} ($(sq3.north east)+(0,3*\ArrDist)$);
  \draw[lbl] ($(sq3.north west)+(-\ArrDist,0)$) -- node[lbltxt,pos=.5,left] {\rotatebox{90}{\relsize{+2}$\pi(N''_{n-i})$}} ($(sq3.south west)+(-\ArrDist,0)$);
  \draw[lblx] ($(s-3-6)+(0,-\ArrDist)$) -- node[lbltxt,pos=1,below] {\relsize{+2}$L'$} ($(s-3-6)+(0,-4*\ArrDist)$);
  \draw[lblx] ($(s-4-6)+(0,-\ArrDist)$) -- node[lbltxt,pos=1,below] {\relsize{+2}$R'$} ($(s-4-6)+(0,-4*\ArrDist)$);
  \draw[lblx] ($(s-6-3)+(\ArrDist,0)$) -- node[lbltxt,pos=1,right] {\relsize{+2}$L''$} ($(s-6-3)+(4*\ArrDist,0)$);
  \draw[lblx] ($(s-6-2)+(\ArrDist,0)$) -- node[lbltxt,pos=1,right] {\relsize{+2}$R''$} ($(s-6-2)+(4*\ArrDist,0)$);
  \draw[lblb] ($(s-6-2)+(2.5*\ArrDist,0)$) -- node[lbltxt,right] {\relsize{+1}$\pi(w'')$} ($(s-6-3)+(2.5*\ArrDist,0)$);
  \draw[lblb] ($(s-4-6)+(0,-2.5*\ArrDist)$) -- node[lbltxt,below] {\relsize{+1}$\pi(w')$} ($(s-3-6)+(0,-2.5*\ArrDist)$);

\end{tikzpicture}}}\\}

\caption{Water filling illustration of the ranking/unranking principle for the $\cdot \OrderRel_r \cdot$ order in product type non-terminals. Each word $w'.w'' \in \Lang{\Nt'_i.\Nt''_{n-i}}$
is uniquely associated with a rectangular compartment of total area $\Pond(w')\cdot\Pond(w'')$.
The ranking of a word $w = w'.w''$ can be adequately compared to the interval on the volume of water (in blue), which upon injection in the matrix, partly fills the compartment associated
with $w$, assuming a water flow in a left-to-right/top-to-bottom lexicographic order. The unranking stage simply consists in searching for the compartment which is partly filled upon injection of a given volume $\rank$.\label{fig:unrankproduct}}
\end{figure}
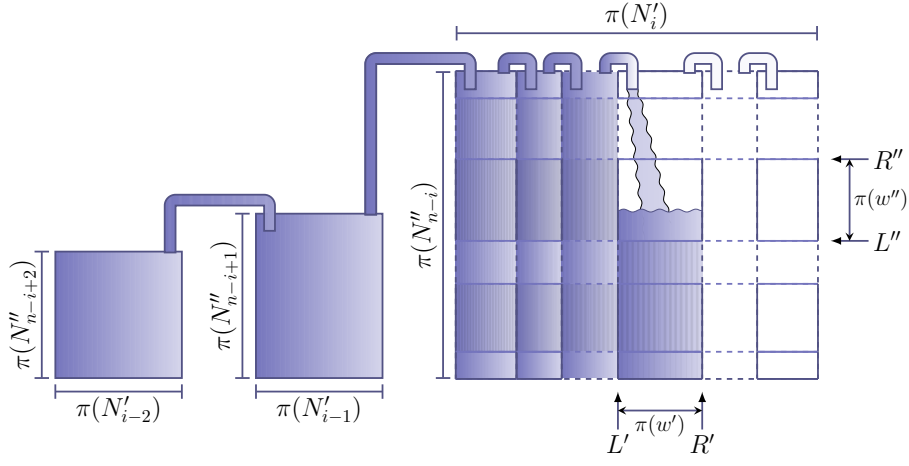

\begin{algorithm}[t]
\caption{Unranking algorithm. Returns a word $w$ and an interval [$I_L$, $I_R$[}
\label{alg:unrank}
{\bf Unrank}$(\Nt_m, \rank)$:
\begin{algorithmic}[5]
\IF[{\bf Union} type] {$\Nt \to \Nt'\;|\;\Nt''$}
  \IF {$\rank< \Pond(\Nt'_m)$}
  	\RETURN {\bf Unrank}$(\Nt'_m, \rank)$
  \ELSE
  	\STATE $(w'', [I_L, I_R[)=${\bf Unrank}$(\Nt''_m, \rank-\Pond(\Nt'_m))$
  	\RETURN $(w'', [I_L+\Pond(\Nt'_m), I_R+\Pond(\Nt'_m)[)$	
  \ENDIF
\ELSIF[{\bf Product} type] {$\Nt \to \Nt'.\Nt''$}
   \STATE$L \leftarrow 0$
    \FORALL{$i\in [1,m-1]$}\label{alg:unranklinebous}
      \IF{$\Pond(\Nt'_i)\cdot \Pond(\Nt''_{m-i}) \le \rank$}
        \STATE $\rank \leftarrow \rank - \Pond(\Nt'_i)\cdot \Pond(\Nt''_{m-i})$
        \STATE $L \leftarrow L + \Pond(\Nt'_i)\cdot \Pond(\Nt''_{m-i})$
      \ELSE[Found the right decomposition]
       \STATE $(w', [L', R'[) = \text{\bf Unrank}(\Nt'_i, \frac{\rank}{\Pond(\Nt''_{m-i})})$
       \STATE $(w'', [L'', R'[) = \text{\bf Unrank}\left(\Nt''_{m-i}, \frac{\rank - L_{\Nt'}\cdot\Pond(\Nt''_{m-i})}{\Pond(w')}\right)$
       \STATE $I_L = L+L'\cdot\Pond(\Nt''_{n-i})+L''\cdot \Pond(w')$
       \STATE $I_R = I_L+\Pond(w')\cdot\Pond(w'')$
       \RETURN $(w'.w'', [I_L, I_R[)$
      \ENDIF
   \ENDFOR
\ELSIF[{\bf Terminal} type] {$\Nt \to \T$}
    \RETURN $(\T,[0,\Pond(\T)])$
\ENDIF
\end{algorithmic}
\end{algorithm}
\subsubsection{Ranking algorithm}\label{sec:ranking}
Let $x\in\mathbb{R}^+$ be a positive real number, and $I=[L,R[\subset \mathbb{R}$ an interval, let us overload the sum
operator through $I+x := [L+x,R+x[$ for the sake of simplicity.
Then an algorithm ${\bf Rank}$ for computing the ranking interval of any word $w\in\Lang{\Nt}_n$ can be outlined as:
\begin{itemize}
 \item {\bf Union type} $\Nt \to \Nt'\;|\;\Nt''$: if $w\in\Nt'_n$ then return ${\bf Rank} (w,\Nt'_n)$.\\
 Otherwise $w\in\Nt''_n$, and return $\Pond(\Nt'_n) + {\bf Rank} (w,\Nt''_n)$.
 \item {\bf Product type} $\Nt \to \Nt'.\Nt''$: Since the grammar is unambiguous, then there only exists one
 decomposition $w = w'.w''$ such that $w'\in\Lang{\Nt'}$ and  $w''\in\Lang{\Nt''}$. Let us then define
 $$ \left[L',R'\right[ \Eqdef {\bf Rank} \left(w',\Nt'_{|w'|}\right) \quad\text{and}\quad  \left[L'',R''\right[ \Eqdef {\bf Rank} \left(w'',\Nt''_{|w''|}\right)$$
 As illustrated by Figure~\ref{fig:unrankproduct}, the returned interval must then be 
$$[L,R[ := \left[ \sum_{i=1}^{|w'|-1} \Pond(\Nt'_i.\Nt'_{n-i}) + L'\cdot\Pond(\Nt''_{n-i})+L''\cdot \Pond(w'), L+\Pond(w')\cdot\Pond(w'') \right].$$
 \item {\bf Terminal type} $\Nt \to\T$: Return $[0,\Pond(\T)[$.
\end{itemize}

\subsubsection{Unranking algorithm}\label{sec:unranking}
Let us now turn to Algorithm~\ref{alg:unrank}, which implements unranking for the relation $\cdot \OrderRel_r \cdot$ and mostly consists in inverting the calculation presented in the Section~\ref{sec:ranking}.

\begin{prop}
  Given a real number $\rank\in [0,\Pond(\Lang{\Gram}_n)[$ Algorithm~\ref{alg:unrank} produces
  the word associated with an interval $I$, $\rank\in I$, in $O(n\log(n))$ arithmetic operations
  after a precomputation in $\Theta(|\NtSet|\cdot n)$ arithmetic operations involving storage of $\Theta(|\NtSet| \cdot n)$ numbers.
\end{prop}
\begin{proof}[Sketch of proof]
First let us outline a proof of correctness by induction for the unranking procedure, starting from
the initial case of terminal rules, where the algorithm returns the only word $\T$, associated with an interval $[0,\Pond(\T)[$.

In the case of union rules, one either need to remove the added contribution $\Pond(\Nt'_m)$ when $\rank\ge\Pond(\Nt'_m)$ before proceeding to unrank within $\Lang{\Nt''_m}$, or directly unrank within $\Lang{\Nt'_m}$ otherwise.

For products rules, one first remarks that $\sum_{i=1}^{|w'|-1} \Pond(\Nt'_i.\Nt'_{n-i})$ is exactly the quantity computed within $L$ in section \ref{sec:ranking}, so one is left to ensure that the remaining part of $\rank$ indeed generates its corresponding word. Namely, let us assume that $w = w'.w'' \in \Lang{\Nt'_i.\Nt''_{n-i}}$,
where $w'$ and $w''$ are associated with intervals $[L',L'+\Pond(w')[$ in $\Lang{\Nt'_i}$ and $[L'',L''+\Pond(w'')[$ in $\Lang{\Nt''_{n-i}}$ respectively.
Therefore the interval associated with $w$ (after subtraction of $L$) is $I=[x,x+\Pond(w')\cdot\Pond(w'')[$ with $x\Eqdef L'\cdot\Pond(\Nt''_{n-i})+L''\cdot \Pond(w')$.
Therefore computing, as done by Algorithm~\ref{alg:unrank}, the quantity $\rank':= \rank/\Pond(\Nt''_{m-i})$ for any $\rank\in I$ gives
$$ L'+\frac{L''}{\Pond(\Nt''_{n-i})}\cdot\Pond(w') \le \rank'< L'+\frac{(L''+\Pond(w''))}{\Pond(\Nt''_{n-i})}\cdot\Pond(w').$$
Since $L''$ is a partial sum of the weights in $\Lang{\Nt''_{n-i}}$, one has $0\le L''\le \Pond(\Nt''_{n-i})-\Pond(w_{|\Lang{\Nt''_{n-i}}|})$
and both bounds are tight (reached by the first and last words). It follows that
$$ L'\le \rank' < L'+\Pond(w') $$
in which one recognizes the interval associated with $w'$ within $\Lang{\Nt'_i}$. The recursive unranking on $\Lang{\Nt''_i}$ is given as
argument
$\rank'' := \frac{\rank - L'\cdot\Pond(\Nt''_{m-i})}{\Pond(w')}$ which, for $\rank\in I$, gives
$$  L''\le \rank''< L''+\Pond(w'')$$
in which one recognizes the interval associated with $w''$ within $\Lang{\Nt''_i}$.
We conclude on the correctness of the algorithm by reminding that the unambiguity of the grammar prevents multiple parsings (i.e. different intervals)
to contribute to the generation of a given word.

The complexity of the algorithm is established by the following observations:
\begin{itemize}
  \item The numbers $\Pond(N_m)$ involved in the unranking procedure can be precomputed thanks to the existence of linear recurrences for the
  coefficients of holonomic generating functions, as discussed in Section~\ref{sec:weights}. 
  They can then be precomputed in $\Theta(n)$ arithmetic operations, requiring storage for $|\NtSet|\cdot\Theta(n)$ large numbers.
  \item The order of investigation of possible decompositions can be modified in Algorithm~\ref{alg:unrank}, line~\ref{alg:unranklinebous}
  to adopt a Boustrophedon order as discussed in Section~\ref{sec:bous}, decreasing the worst-case complexity of the algorithm from $\BigO{n^2}$ to $\BigO{n\log(n)}$ in the worst-case. 
  The total ordering on words can then be redefined to account for such a change, and the proof of correctness is easily adapted.
\end{itemize}
\end{proof}

\subsection{Random generation of numbers in gapped intervals} \label{sec:modRand}

\begin{figure}
  {\centering \resizebox{\textwidth}{!}{{\begin{tikzpicture}[inner sep=0]
  \definecolor{rougeForb}{HTML}{eb23238f}
  \definecolor{rougeForbP}{HTML}{6d1515ff}
  \newcommand{\MyScale}{12}
  \newcommand{\Height}{1.3}
  \newcommand{\HDist}{8}
  \newcommand{\ShiftB}{0}
  \tikzstyle{okbox}=[draw=blue]
  \tikzstyle{corr}=[dashed]
  \tikzstyle{forbidden}=[draw=rougeForbP,fill=rougeForb,line width=1.2pt]

\foreach \i/\ival in {0/0,1/7,2/8.4,3/11,4/15,5/17,6/23,7/27,8/31}
{
  \node at ($(\MyScale*\ival pt,0)$) (s\i) {};
  \node[below=\MyScale*\Height pt of s\i] (t\i) {};
  \draw (s\i.center) -- (t\i.center);
}

\foreach \i/\ival in {0/0,1/7,2/10,3/12,4/18,5/22}
{
  \node[below=\MyScale*\HDist pt of t\i,xshift=\MyScale*\ShiftB pt] at ($(\MyScale*\ival pt,)$) (u\i) {};
  \node[below=\MyScale*\Height pt of u\i] (v\i) {};
  \draw (u\i.center) -- (v\i.center); 
}

\path[forbidden]   (s1.center) to (t1.center) to[out=-90,in=-90] (t2.center) -- (s2.center) to[out=90,in=90]
(s3.center) -- (t3.center) to[out=-90,in=-90] (t4.center) -- (s4.center) to[out=90,in=90,looseness=1.3]
(s6.center) -- (t6.center) to[out=-90,in=-90] (t7.center) -- (s7.center) to[out=90,in=90,looseness=.6]
(s1.center);

\foreach \i/\ival in {0/0,1/1,1/2,2/3,2/4,3/5,4/6,4/7,5/8}
{
  \draw[corr] (t\ival.center) -- (u\i.center); 
}

\foreach \i in {0,...,7}
{
  \pgfmathtruncatemacro\k{\i+1}
  \draw[fill=white,opacity=.5] (s\i.center) -- (t\i.center) -- (t\k.center) -- (s\k.center)  -- cycle;
  \draw[] (s\i.center) -- (t\i.center) -- (t\k.center) -- (s\k.center)  -- cycle;
}

\foreach \i in {0,...,4}
{
  \pgfmathtruncatemacro\k{\i+1}
  \draw (u\i.center) -- (v\i.center) -- (v\k.center) -- (u\k.center)  -- cycle;
}

\node[above=10pt of s0] {\relsize{2}$0$};
\node[above=10pt of s8] {\relsize{2}$\pi(\mathcal{L})$};
\node[below=10pt of v0] {\relsize{2}$0$};
\node[below=10pt of v5] {\relsize{2}$\pi(\mathcal{L}/\mathcal{F})$};

\foreach \i in {0,...,7}
{
  \pgfmathtruncatemacro\k{\i+1}
  \path[draw=none] (s\i.center) -- node[pos=.5] (tmpr1) {} (t\i.center);
  \path[draw=none] (s\k.center) -- node[pos=.5] (tmpr2) {} (t\k.center);
  
  \path[draw=none] (tmpr1) -- node[pos=.5] {$w_\k$} (tmpr2);
}

\foreach \i/\ival in {0/1,1/3,2/5,3/6,4/8}
{
  \pgfmathtruncatemacro\k{\i+1}
  \path[draw=none] (u\i.center) -- node[pos=.5] (tmpr1) {} (v\i.center);
  \path[draw=none] (u\k.center) -- node[pos=.5] (tmpr2) {} (v\k.center);
  
  \path[draw=none] (tmpr1) -- node[pos=.5] {$w_\ival$} (tmpr2);
}

\node at ($(\MyScale*13 pt,\MyScale*1.5 pt)$) {\relsize{+4}$\mathcal{F}$};

  \path[draw=none] (s5.center) -- node[pos=.5] (tmpr1) {} (t5.center);
  \path[draw=none] (s6.center) -- node[pos=.5] (tmpr2) {} (t6.center);
  \path[draw=none] (tmpr1) -- node[pos=.2] (tmpr4) {} (tmpr2);
  \path[draw] ($(tmpr4)+(3pt,3pt)$) -- ($(tmpr4)+(-3pt,-3pt)$);
  \path[draw] ($(tmpr4)+(-3pt,3pt)$) -- ($(tmpr4)+(3pt,-3pt)$);

  \path[draw=none] (u3.center) -- node[pos=.5] (tmpr1) {} (v3.center);
  \path[draw=none] (u4.center) -- node[pos=.5] (tmpr2) {} (v4.center);
  \path[draw=none] (tmpr1) -- node[pos=.2] (tmpr5) {} (tmpr2);
  \path[draw] ($(tmpr5)+(3pt,3pt)$) -- ($(tmpr5)+(-3pt,-3pt)$);
  \path[draw] ($(tmpr5)+(-3pt,3pt)$) -- ($(tmpr5)+(3pt,-3pt)$);

  \path[draw,-latex,line width=1pt] (tmpr5) -- (tmpr4);
  \path[draw,-latex,thick] ($(tmpr5)+(0,-10pt)$) -- node[pos=0,below, inner sep=2pt] {$r:= Rand([0,\pi(\mathcal{L}/\mathcal{F})[)$} (tmpr5);

  \path[draw,-latex,thick] ($(tmpr4)+(0,10pt)$) -- node[pos=0,above, inner sep=2pt,text centered,text width=8.7em, anchor=238] 
{$r':= r+\delta$\\$\delta:= \pi(w_2)+\pi(w_4)$} (tmpr4);

\end{tikzpicture}   
}} \\}

  \caption{Illustration of the shift function $\Shift{}$. In order to avoid any forbidden words, one needs to \emph{shift rightward} a random number $\rank$ by the total weight of
  forbidden words (red area) that are found \emph{leftward}.}\label{fig:shift}
\end{figure}
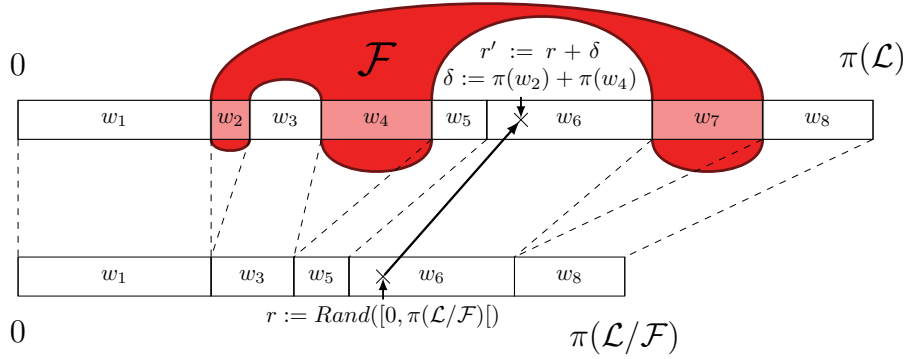
In the previous section, a simple weighted unranking algorithm was proposed.  Therefore by generating a random number $r$ uniformly  in $[0,\Pond(\mathcal{L})[$, and using the {\bf Unranking} algorithm, a word $w$ can be generated with respect to the weighted distribution over a language $\mathcal{L}$. However when a forbidden set $\Forb$ is given, one additionally needs to avoid any interval associated with a forbidden word. In other words, one can no longer draw a random number uniformly in $[0,\Pond(\mathcal{L})[$, but rather in $$I_{\bar\Forb} := \cup_{w\notin\Forb}I_w = [0,\Pond(\mathcal{L})[\backslash (\cup_{{\bar w}\in \Forb} I_{{\bar w}}).$$

Since the intervals $I_w$ are mutually disjoint subsets of $[0,\Pond(\mathcal{L})[$, a possible strategy consists in drawing a random number $\rank\in[0,\Pond(\mathcal{L})-\Pond(\Forb)[$, and increment $\rank$ by some quantity $\Shift{\rank,\Forb}$ that sends $\rank+\Shift{\rank,\Forb}$ into $I_{\bar \Forb}$. Considering Figure~\ref{fig:shift}, one observes that $\Shift{\rank,\Forb}$ can be inductively defined as the total weight
of all forbidden words smaller than the word found at $\rank+\Shift{\rank,\Forb}$. In general, one could order the forbidden words in $\Forb$ and traverse $\Forb$ to compute $\Shift{\rank,\Forb}$, but
this would induce spending an additional $\mathcal{O}(|\Forb|)$ arithmetic operations per generation.
For this reason, the intervals of forbidden words are gathered in a balanced binary tree structure that grants access to $\Shift{\rank,\Forb}$
in $\mathcal{O}(\log(|\Forb|))$ operations.

\subsubsection{AVL tree for forbidden intervals}
\begin{algorithm}[t]
\caption{\modRandom: Takes a uniform random number and a node, and returns a uniform random number that avoids any interval associated with already generated words.}
\label{modRandom}

$\modRandom(\rank, v)$
\begin{algorithmic}[5]

\IF {$v= \varnothing$}
	\RETURN $\rank$
\ENDIF
\STATE $(\LeftTree,\RightTree,{\bar w}, [L_{{\bar w}},R_{{\bar w}}[,\Mass{{\bar w}})\gets v$
\IF {$\rank < L_{{\bar w}}-\Mass{{\bar w}}$}
	\STATE $\modRandom(\rank, \LeftTree)$
\ELSE
	\STATE $\modRandom(\rank+\Mass{{\bar w}_i}+\Pond({\bar w}_i), \RightTree)$
\ENDIF
\end{algorithmic}
\end{algorithm}

For each (word, interval) pair produced by the unrank algorithm, a corresponding node is inserted into an AVL tree~\cite{Adelson-Velskii1962}, i.e. a self-balancing binary search tree, whose height after $k$ insertions 
can be limited to $\Theta(\log(k))$ through balancing operations.
Since the intervals associated with the forbidden set are non overlapping, then they can be compared and therefore
stored within an AVL tree. It follows that the insertion and lookup of $k$ intervals can be performed in $\Theta(k\log(k))$ comparisons
in the worst-case scenario.

Let us then define recursively our tree as either the empty tree, denoted by $\varnothing$, or a 5-tuple $v=(\LeftTree,\RightTree,{\bar w}, I_{{\bar w}},\Mass{{\bar w}})$ where:
\begin{itemize}
  \item $\LeftTree$ and $\RightTree$ are respectively the left and right children of the tree.
  Both can possibly be empty trees.
  \item ${\bar w}$ and $I_{{\bar w}}:=[L_{{\bar w}},R_{{\bar w}}[$ are a forbidden word and its corresponding interval.
  \item $\Mass{{\bar w}}$ is the total weight of forbidden intervals in the left subtree.
\end{itemize}
Let us remind that the nodes of an AVL tree are such that any node in a left subtree is less than or equal to its
root, itself being less than or equal to any node of its right subtree.
Also let us remark that, upon inserting in a tree $v_{\bar w}$
a new word ${\bar w'}\neq{\bar w}$ associated with an interval $I_{\bar w'}=[L_{\bar w'},R_{\bar w'}[)$,
the value $\Mass{{\bar w}}$, initialized at  $0$, can be easily updated into a new value $\Mass{{\bar w}}'$ such that
\begin{equation}
  \Mass{{\bar w}}'  = \left\{ \begin{array}{cl} \Mass{{\bar w}}+ \Pond({\bar w})&
  \text{If }{\bar w}'\OrderRel_r {\bar w} \text{, i.e. }{\bar w}' \text{ is inserted in the left subtree }\LeftTree\text{ of }v \\ \Mass{{\bar w}} &  \text{Otherwise} \end{array} \right.
  \label{eq:update}
\end{equation}

Assuming the tree is correctly built, Algorithm~\ref{modRandom} simply descends into the tree, and computes $\Shift{\rank,\Forb}$
incrementally. For a given node $v=(\LeftTree,\RightTree,{\bar w}, I_{{\bar w}},\Mass{{\bar w}})$, the algorithm determines if $\rank$
corresponds to a word in the interval covered by $\LeftTree$,
by comparing $\rank$ to $L_{\bar w}-\Mass{{\bar w}}$ the total mass of allowed words in $\LeftTree$. If smaller, then $\rank$ remains unmodified
and the algorithm is run recursively on $\LeftTree$. If greater, then the final interval reached by $\rank$ is greater than $I_{\bar w}$, and
fits in the right subtree $\RightTree$. The value $\rank$ is then incremented by the total mass $\Mass{{\bar w}}+\Pond({\bar w})$ of forbidden words smaller than $\RightTree$,
and this value is used within a recursive call on $\RightTree$.
This process is terminated when the empty tree $\varnothing$ is reached, and the current value of $\rank$ is returned.
In other words, the returned value $\rank$ is distant from its original value by the sum of weights $\Mass{{\bar w}}$ on the left subtrees whose
intervals are dominated by $\rank$, in which one recognizes the definition of $\Shift{\rank,\Forb}$.

\subsubsection{Correctness}
\begin{prop} \label{propModRandom}
The function $\modRandom$ computed by Algorithm~\ref{modRandom} is a bijection from $[0, \Pond(\mathcal{L})-\Pond(\Forb))[$
onto $[0,\Pond(\mathcal{L})[\backslash (\cup_{{\bar w}\in \Forb} I_{{\bar w}})$ with uniform density.
  \end{prop}
\begin{proof}

The outline of the proof is as follows: First we establish a technical invariant on the subset of values
passed to Algorithm~\ref{modRandom}. Using this invariant, we show that the final value returned by $\modRandom$
avoids every forbidden interval, and that any interval can be reached.
Let us start with some notations, followed by a technical lemma.

Let $v_i$ be the $i$-th node in the tree and let us denote by $[a,\ldots,i,\ldots,b]$ the indices of nodes accessible from $v_i$.
Then let us denote by $H_{i}$ the interval that is \emph{dominated} by  $v_i$, defined as
$$ H_i = \left[ R_{{\bar w}_{a-1}}, L_{{\bar w}_{b+1}} - \sum_{k= a}^b \Pond({\bar w}_k)\right[$$
where $R_{{\bar w}_{i}}, i\in[1,|\Forb|],$ the upper bound  (resp. $L_{{\bar w}_{i}}, i\in[1,|\Forb|],$ the lower bound) of the forbidden interval of index $i$ is extended
by $R_{{\bar w}_{0}} = 0$ (resp. $L_{{\bar w}_{|\Forb|}}=\Pond(\La)$).

\begin{lem}\label{lem:int}
  Let $v_i=(\LeftTree,\RightTree,{\bar w}, [L_{{\bar w}_i},R_{{\bar w}_i}[,\Mass{{\bar w}})$ be a node in the tree.
  Then the set of values $\rank$ passed as argument to $\modRandom$ jointly with $v_i$ is exactly $H_i$.
\end{lem}
\begin{proof}
Let us prove this claim by induction on the depth $D$ of recursive calls.
Clearly in the initial call ($D=0$), $v_i$ is the root node and $H_i$ is the whole
interval $[0,\Pond(\La)-\Pond(\Forb)[$ from which $\rank$ is drawn uniformly, so our claim holds.
Assume now that the set of possible values for $\rank$ is exactly $H_i := [ R_{{\bar w}_{a-1}}, L_{{\bar w}_{b+1}} - \sum_{k= a}^b \Pond({\bar w}_k)[$ at a given depth $D=M$,
then let us investigate the recursive calls.
Two cases arise, depending on the value of $\rank$:
\begin{itemize}
\item When $\rank\in \mathcal{A}=[R_{{\bar w}_{a-1}},L_{{\bar w}_i}-\Mass{{\bar w}_i}[$, then $\modRandom$ is called on $v_j:=\LeftTree$
  with unmodified value $\rank':=\rank$. Thanks to the binary search tree structure, the indices of the forbidden nodes on the left subtree are $[a,\ldots,i-1]$,
  and $H_j = [ R_{{\bar w}_{a-1}}, L_{{\bar w}_{i}} - \sum_{k=a}^{i-1} \Pond({\bar w}_k)[$.
  Since $\Mass{{\bar w}_i}=\sum_{k=a}^{i-1} \Pond({\bar w}_k)$ (def.), then $H_j = \mathcal{A}$,
  and any value $\rank'\in H_j$ can therefore be passed to the subsequent call.
\item When $\rank\in \mathcal{B}=[L_{{\bar w}_i}-\Mass{{\bar w}_i},L_{{\bar w}_{b+1}} - \sum_{k= a}^b \Pond({\bar w}_k)[$, then $\modRandom$
  is called on $v_j:=\RightTree$ with value $\rank' := \rank+\Mass{{\bar w}_i}+\Pond({\bar w}_i)$.
  The indices of the forbidden nodes on the right subtree are $[i+1,\ldots,b]$, so one has $H_j = [R_{{\bar w}_i}, L_{{\bar w}_{b+1}} - \sum_{k= i+1}^{b} \Pond({\bar w}_k)[$.
  The image $\mathcal{B'}$ of the interval $\mathcal{B}$ through a shift of value $\Mass{{\bar w}_i}+\Pond({\bar w}_i)$ is then
  \begin{align*}
    \mathcal{B'} &= \left[L_{{\bar w}_i}+\Pond({\bar w}_i) ,L_{{\bar w}_{b+1}} - \sum_{k= a}^b \Pond({\bar w}_k) + \sum_{k=a}^{i-1} \Pond({\bar w}_k)+\Pond({\bar w}_i) \right[ \\
    &= \left[R_{{\bar w}_i}, L_{{\bar w}_{b+1}} - \sum_{k= i+1}^{b} \Pond({\bar w}_k)\right[   = H_j.
  \end{align*}
  Finally, since $\rank$ can be any value in $\mathcal{B}$, then any value $\rank'\in H_j$ can be passed to $\modRandom$ for some value $\rank\in\mathcal{B}$.
\end{itemize}
Consequently at depth $D=M+1$, the values $r'$ provided to $\modRandom$ over a subtree $v_j$ are exactly $H_j$, and this property therefore holds for any $D\ge0$.
\end{proof}

Let us show that forbidden intervals are indeed avoided.
Let us consider a node $v_i=(\LeftTree,\RightTree,{\bar w}, [L_{{\bar w}},R_{{\bar w}}[,\Mass{{\bar w}})$, giving rise to a call  $\modRandom(\rank',\varnothing)$,
itself returning the final value.
Since, for this node, Lemma~\ref{lem:int} holds, then the value passed to this call is any $r\in[R_{{\bar w}_{i-1}}, L_{{\bar w}_{i+1}}-\Pond({\bar w}_{i})[$.
 Therefore either $\rank<L_{{\bar w}_i}$ and $\rank\in[R_{{\bar w}_{i-1}},L_{{\bar w}_i}[$ is returned,
or $\rank\ge L_{{\bar w}_i}$ and $\rank+\Pond({\bar w}_{i})\in [R_{{\bar w}_i},L_{{\bar w}_{i+1}}[$ is returned. It follows that any returned value $r'$
falls between two consecutive forbidden intervals (resp. within the ending intervals $[0,L_{{\bar w}_1}[$ or $[R_{{\bar w}_{|\Forb|}},\Pond(\La)[$),
and therefore cannot fall in a forbidden interval.

Furthermore let us prove that any two calls $\modRandom(\rank',\varnothing)$ and $\modRandom(\rank'',\varnothing)$ from $v_i$ and $v_j$ respectively, $i\neq j$,
give rise to distinct intervals.
Recall that, as pointed out in the previous paragraph, the possibly generated intervals from a node $v_i$ are $[R_{{\bar w}_{i-1}},L_{{\bar w}_i}[$ if $\LeftTree=\varnothing$,
and $[R_{{\bar w}_i},L_{{\bar w}_{i+1}}[$ if $\RightTree=\varnothing$.
Therefore, by contradiction, any two calls giving rise to similar intervals would have to involve consecutive nodes $v_i$ and $v_{i+1}$ such that the right subtree $v_i$ is
$\RightTree_i=\varnothing$ and the
left subtree of $v_{i+1}$ is $\LeftTree_{i+1}=\varnothing$. Since such two nodes would represent consecutive values, then one would appears in a subtree of the other,
otherwise the first common ancestor $v_j$ of $v_i$ and $v_{i+1}$ would be such that $v_i<v_j<v_i+1$ and the two nodes would not be consecutive.
Since $v_i<v_{i+1}$, then either $v_i$ would be found in the left subtree
of $v_{i+1}$ (and then $\LeftTree_{i+1}\neq\varnothing$), or $v_{i+1}$ would be found in the right subtree of $v_{i}$ (and then $\RightTree_{i}\neq\varnothing$).
Both situations contradict the premisses, thus any interval $[R_{{\bar w}_{i-1}},L_{{\bar w}_i}[, i\in[1,|\Forb|+1]$ is generated by at most a call over a single empty tree node $\varnothing$.

We conclude with the remark that there are exactly $|\Forb|+1$ leaves in a binary tree with $|\Forb|$ inner nodes. Since there are also $|\Forb|+1$ intervals
$[R_{{\bar w}_{i-1}},L_{{\bar w}_i}[, i\in[1,|\Forb|+1]$ which are generated by at most one leave, then any such interval is generated, and $\modRandom$
is therefore a bijection of $[0,\Pond(\La)-\Pond(\Forb)[$ into $\cup_{i=1}^{|\Forb|+1}[R_{{\bar w}_{i-1}},L_{{\bar w}_i}[ = [0,\Pond(\La)[\backslash(\cup_{i=1}^{|\Forb|} I_{{\bar w}_i})$.

Finally, since the map $\modRandom$ involves only shifts and no scaling, it follows that the map is measure preserving.  Thus the algorithm alters uniformly generated random numbers
over $[0,\Pond(\La)-\Pond(\Forb)[$ into uniform random numbers over $[0,\Pond(\La)[\backslash(\cup_{i=1}^{|\Forb|} I_{{\bar w}_i})$.
\end{proof}

\subsubsection{Complexity considerations}

  As can be seen in Equation~\ref{eq:update}, updating the values $\Mass{v}$ in a tree with $m$ nodes can be done in $\mathcal{O}(\log(m))$ arithmetic operations upon insertion of a new node.
  However the AVL structure also requires a post-processing consisting of $\mathcal{O}(\log(m))$ \emph{shifts} to keep the tree balanced.
  The shift operation involves taking two nodes $v_i<v_j$ that are connected in the tree and switching
  their ancestrality. Namely, if $v_i$ was the first node of the left subtree of $v_j$, then $v_j$ becomes
  become the first node of the right subtree of $v_i$ (and vice-versa). The effect of this operation is
  local, therefore in any pair $(v_i,v_j)$ of nodes inverted by a shift operation, the values $\Mass{v_i}$
  and $\Mass{v_j}$ can be updated in $\mathcal{O}(1)$ arithmetic operations, and  the overall
  cost of $k$ insertions remains in $\mathcal{O}(k\log(k))$ arithmetic operations.

  Each internal node maintains a possibly large number $\Mass{}$, therefore $\Theta(|\Forb|)$ numbers need be stored in the tree.
  The ratio of probability between the most and least probable structure grows like $\Omega(\alpha^n)$, therefore
  at least $\Theta(n)$ bits needs be used for the numbers.

\section{Conclusion and perspectives}\label{sec:conclusion}

	We addressed the random generation of non-redundant sets of sequences from context-free languages,
	while avoiding a predefined set of words. We first investigated the efficiency of a rejection
	approach. Such an approach was found to be acceptable in the uniform case. By contrast, for weighted
	languages, we showed that for some languages the expected number of rejections would grow exponentially
	on the desired number of generated sequences for at least two parameters. Furthermore, we showed that in typical context-free languages and 	for fixed length, the probability distribution can be dominated by a small number of sequences.
	We proposed a first algorithm for this problem, based on the recursive approach. The correctness of the
	algorithm was demonstrated, and its efficient implementation discussed. This algorithm was showed to perform
	a non-redundant generation of $k$ distinct structures in $\mathcal{O}(k\cdot n\log(n))$, after a precomputation in $\Theta(n\cdot |\NtSet|)$ arithmetic operations,
    and requires storage of $\mathcal{O}(n\cdot |\NtSet|+k)$ large numbers, and a data structure consisting of $\Theta(n\cdot k)$ nodes.
  We explored a second approach, based on a ranking/unranking approach for the same task, and obtained an
algorithm in  $\mathcal{O}(n\cdot |\NtSet|+k\cdot n \log n)$ complexity, with the slightly decreased memory consumption of $\mathcal{O}(n\cdot |\NtSet|+k)$ large numbers.
	These complexities hold in the worst-case scenario, and remain mostly unaffected by the magnitude of weights being used.

	\subsection{Different impact of fixed-precision arithmetics implementations} 
When using arbitrary (or sufficient) precision arithmetics, the complexity and storage of the two algorithms are the same.
However, practical implementations may involve using fixed-precision arithmetic, in which case significant differences between the two methods arise.
The complexity of both algorithms can be improved significantly if one uses fixed-precision arithmetic.  However, in both cases, the algorithms suffer from a quantifiable loss of precision.  

If the ratio between the weight of the smallest word and the weight of the language is small, then built-in floating point operations may be used, giving some advantages to the unranking approach with respect to its memory consumption. Indeed, the cost of storing the data structure will then dominate the memory consumption of the recursive version ($\mathcal{O}(n\cdot |\NtSet|+n\cdot k)$), while the memory complexity of the unranking algorithm gently decreases to ($\mathcal{O}(n\cdot |\NtSet|+k)$).

However, we believe the recursive method to be more stable numerically than the unranking approach.
Indeed, the weights accessible on the alternative choice in the usual generation are typically comparable.
Therefore, it will typically take an large number of generations for the recursive algorithm to fully deplete one of the alternatives. By contrast, the unranking algorithm may very quickly isolate a poorly contributing set of words after very few generation. For instance, if the second word in the ordering is generated first, then the data structure may practically forbid the first element, choosing it with $0$ probability because of the rounding error. This point therefore seems favorable to the recursive algorithm.

\subsection{Perspectives}
  Let us briefly outline a few perspectives to the current work:
  \begin{itemize}
	\item {\bf Decomposable structures:} One natural extension of the current work concerns the random generation of the more general class of decomposable
	structures~\cite{flajoletcalculus}. Indeed, such aspects like the \emph{pointing} and \emph{unpointing} operator 	are not explicitly accounted for in the current work. Furthermore, the generation of labeled structures might
	be amenable to similar techniques in order to avoid a redundant generation. It is unclear, however, how
	one may extend the notion of parse tree in this context. Intrinsic ambiguity issues might arise, for instance while using 	the \emph{unranking} operator.

    \item {\bf Non-redundant Boltzmann sampling.} Another direction for an efficient implementation of the non-redundant generation may rely on an extension of Boltzmann
    samplers~\cite{fullboltz}. Indeed, the prefix-tree introduced by the step-by-step
    algorithm could, in principle, be used \emph{as is} to correct the probabilities used by Boltzmann 
  		sampling. However, it is unclear how such a correction may impact the probability of rejection, and 	
		 consequently degrade the performances of the resulting algorithm.

    \item {\bf Accommodating general sets of forbidden words.}
      Both the step-by-step and unranking algorithms require the preliminary insertion of the forbidden set $\Forb$ into a dedicated data structure (prefix tree/AVL tree),
      both requiring the parse trees/walks of any word in $\Forb$ to be available.
      When such an information is not available, one could in principle parse the words in $\Forb$ to
      build the tree. In general this may require $\Forb$ run of a $n^{3-\varepsilon}$ parsing algorithm, leading to an
      impractical $\mathcal{O}(n^{3-\varepsilon}\cdot|\Forb|)$ complexity. In practice, it seems more fruitful to simply
      run the algorithm starting from an empty tree, and to test after each generation
      if the generated word is found in $\Forb$. If so, reject it after adding its parse walk, available to the algorithm
      without further computation since the word was just created, to the tree. Since this update is made at most once for
      each word in $\Forb$, then the worst-case complexity of generating $k$ words remains bounded by $\mathcal{O}(|\Forb|\cdot n\log(n))$ arithmetic operations.

  \end{itemize}
\section*{Acknowledgements}
	 The author would like to thank A. Denise, C. Herrbach, and an anonymous reviewer
  for thought-provoking remarks and helpful suggestions. This work was supported by the French \emph{Agence Nationale de la Recherche} 
as part of the {\sc Magnum} project ({\tt ANR 2010 BLAN 0204}).

\bibliographystyle{amsplain}
\bibliography{biblio}

\newpage
\appendix

\section{Expressivity of the binary Chomsky normal form}

 Let us show that the assumption of a BCNF can be made without loss of generality (or performance).
Indeed, it is a classic result that any context-free grammar $\Gram$ can be transformed into a Chomsky Normal Form (CNF) grammar that generates the same language. 
\subsection{From CNF to BCNF grammars: An algorithm}
From such a grammar, an equivalent grammar in BCNF can be simply and efficiently obtained through the following transformation: 
\begin{enumerate}[i)]
\item For each terminal $\T$ (resp. empty word $\varepsilon$) create a new non-terminal $\Nt_\T$ (resp. $\Nt_\varepsilon$) whose sole production is $\Nt_\T \Production \T$ (resp. $\Nt_\T \Production \varepsilon$); 
\item  Replace any occurrence of $\T$ (resp. $\varepsilon$) within a production rule with its dedicated non-terminal $\Nt_\T$ (resp. $\Nt_\varepsilon$);
\item Replace any rule $\Nt \Production \Nt'.\Nt''$, where $\Nt$ has more than one derivation, with rules
$\Nt \Production \Nt^\bullet$ and $\Nt^\bullet \Production \Nt'.\Nt''$, where $\Nt^\bullet$ is a newly created non-terminal; 
\item For any non-terminal $\Nt$ having multiple production rules $(\Nt\to \UnG{X_1}{\UnG{\cdots}{X_k}}$,\;$ k>1$), create $k-2$ dedicated non-terminals $\{\Nt_i\}_{i=1}^{k-2}$, and replace the rules of $\Nt$ with a tree-like equivalent hierarchy of binary rules. For instance, one may create chained rules, such that $\Nt \Production \UnG{X_1}{\Nt_1}$, $\{\Nt_i \Production \UnG{X_{i+1}}{\Nt_{i+1}}\}_{i=1}^{k-3}$, and $\Nt_{k-2} \Production \UnG{X_{k-1}}{X_{k}}$;
\item Finally, remove every non-terminal whose sole production is $\Nt \Production \Nt'$, replacing any occurrence of $\Nt$ by $\Nt'$ in any derivation rule. 
\end{enumerate}

\subsection{Correctness}
The equivalence of the resulting grammar to the input one in CNF trivially follows from the language-preserving nature of the substitutions performed at each step. Furthermore, it is easily verified that the resulting grammar is in BCNF. Indeed, consider the set $\ProdRules_\Nt$ of derivation rules available for any former non-terminal $\Nt$,  along the transformation: 
\begin{itemize}
\item Before executing the transformation:  $\ProdRules_\Nt$  consists an arbitrary number of terminal rules ($\Nt \to \Nt'$), binary-product ($\Nt \to \PrG{\Nt'}{\Nt''}$) rules, or possibly an epsilon rule for the axiom ($\Axiom \to \varepsilon$);
\item Steps i) and ii) remove terminal symbols:  After their execution, $\ProdRules_\Nt$  contains an arbitrary number of  unary ($\Nt \to \Nt'$) or binary ($\Nt \to \PrG{\Nt'}{\Nt''}$) rules;
\item Step iii) removes non-unary multiple rules: $\ProdRules_\Nt = \{(\Nt \to \PrG{\Nt'}{\Nt''})\}$, 
or $\ProdRules_\Nt = \{\Nt \to \Nt'\;|\;\Nt' \in \NtSet' \subseteq \NtSet\}$; 
\item Step iv) binarizes multiple rules:  $\ProdRules_\Nt = \{(\Nt \to \Nt')\}$, 
 $\ProdRules_\Nt = \{(\Nt \to \PrG{\Nt'}{\Nt''})\}$,
or $\ProdRules_\Nt = \{(\Nt \to \Nt'),(\Nt \to \Nt'')\}$,
where $\Nt',\Nt'' \in \NtSet$;
\item Finally, step v) removes extraneous unary non-terminals:  $\ProdRules_\Nt = \{(\Nt \to \PrG{\Nt'}{\Nt''})\}$,
or $\ProdRules_\Nt = \{(\Nt \to \Nt'),(\Nt \to \Nt'')\}$,
$\Nt',\Nt'' \in \NtSet$.
\end{itemize}
The derivation rules available for the set of non-terminals, created during the transformation, are initially in BCNF. Note that the only modification performed on productions of new non-terminals substitute a non-terminal for another, thereby keeping the rules BCNF-compliant. Finally, the constraint on the initial CNF guarantees that the only epsilon rule is derived from the axiom, either in a single production or through a sequence of non-referential productions created at step iv).
Therefore one concludes that the produced grammar is indeed in BCNF.

The proposed transformation from a CNF to an equivalent BCNF can be implemented in linear time, through a careful ordering of the removals performed by step v), and the number of rules is at most increased by a constant factor.

\end{document}